\documentclass[12pt]{article}
\usepackage{amsmath}
\usepackage{enumerate}
\usepackage{natbib}
\usepackage{url} 
\usepackage{comment}

\newcommand{\blind}{1}

\usepackage[toc,page]{appendix}
\usepackage{amsmath}
\usepackage{multirow}

\usepackage{graphicx}

\usepackage{arydshln}
\usepackage{amssymb}
\usepackage{float}
\usepackage{subcaption}
\usepackage{rotating}  
\usepackage{csquotes}

\newtheorem{proposition}{Proposition}
\newenvironment{proof}[1][Proof]{\noindent \textbf{#1.} }{\  \rule{0.5em}{0.5em}}

\newcommand{\E}{\mathrm{E}}

\allowdisplaybreaks

\begin{document}

\def\spacingset#1{\renewcommand{\baselinestretch}%
{#1}\small\normalsize} \spacingset{1}


\if1\blind
{
  \title{\bf Why you should also use OLS estimation of tail exponents}

 \author{Thiago Trafane Oliveira Santos\hspace{.2cm}\\
    Central Bank of Brazil, Brasília, Brazil\\ Department of Economics, University of Brasilia, Brazil.\\ Email: thiago.trafane@bcb.gov.br\\
    and \\
    Daniel Oliveira Cajueiro\thanks{DOC thanks to CNPq for partial financial support (304706/2023-0).} \\
    Department of Economics, University of Brasilia, Brazil.\\  National Institute of Science and Technology for Complex Systems (INCT-SC)\\Machine Learning Laboratory in Finance and Organizations (LAMFO), Brazil.\\ Email: danielcajueiro@gmail.com}
  \maketitle
} \fi

\if1\blind
{
  \bigskip
  \bigskip
  \bigskip
  \begin{center}
\end{center}
  \medskip
} \fi

\bigskip
\begin{abstract}
			Even though practitioners often estimate Pareto exponents running OLS rank-size regressions, the usual recommendation is to use the Hill MLE with a small-sample correction instead, due to its unbiasedness and efficiency. In this paper, we advocate that you should also apply OLS in empirical applications. On the one hand, we demonstrate that, with a small-sample correction, the OLS estimator is also unbiased. On the other hand, we show that the MLE assigns significantly greater weight to smaller observations. This suggests that the OLS estimator may outperform the MLE in cases where the distribution is (i) strictly Pareto but only in the upper tail or (ii) regularly varying rather than strictly Pareto. We substantiate our theoretical findings with Monte Carlo simulations and real-world applications, demonstrating the practical relevance of the OLS method in estimating tail exponents.\end{abstract}

\noindent%
\vfill

\newpage
\spacingset{1.8} 

	\section{Introduction} \label{sec:intro}	

	Numerous studies have found empirical evidence supporting a Pareto distribution for various economic and financial variables, at least in their upper tails. Notable examples include personal income and wealth (e.g., \citealt{pareto1896cours} and \citealt{klass_etal2006forbes}), population of cities (e.g., \citealt{gabaix1999zipf} and \citealt{gabaix_ioannides2004evolution}), size of firms (e.g., \citealt{okuyama_etal1999zipf}, \citealt{axtell2001zipf}, \citealt{fujiwara_etal2004pareto}, \citealt{luttmer2007selection}, \citealt{gabaix_landier2008why}, \citealt{digiovanni_etal2011power}, and \citealt{digiovanni_levchenko2013firm}), and financial returns and trading volume (e.g., \citealt{gopikrishnan_etal1999scaling} and \citealt{gopikrishnan_etal2000statistical}). 

	If a random variable $S$, $S\geq \underline{s}>0$, is Pareto distributed with a tail exponent $k>0$, its survival function is given by
	\begin{align} 
		P(S\geq s) =  & \left(\underline{s}/s\right)^{k} \label{eq:sf_pareto}
	\end{align}
	for $s \geq \underline{s}$. The Pareto distribution is characterized by its fat tails, and the tails are heavier for lower values of the tail exponent $k$. Specifically, the $m$-th moment exists if and only if $k>m$, meaning no moment exists for $k \in (0,1]$. This characteristic underscores the importance of accurately estimating the Pareto tail exponent. Practitioners often estimate it using OLS rank-size regressions. However, the prevailing recommendation in the literature (e.g., \citealt{clauset_etal2009power}) is to use the maximum likelihood estimator (MLE) proposed by \cite{hill1975simple}, with a small-sample correction. This preference is due to the MLE's properties of unbiasedness and efficiency, as it is the best linear unbiased estimator (BLUE) and also the minimum variance unbiased estimator (MVUE) (\citealt{aban_meerschaert2004generalized}).
	
	In this paper, we advocate for the use of OLS estimation of tail exponents in addition to MLE. Our reasoning is based on two theoretical results. First, we establish that the OLS estimator of $d \equiv 1/k$ is also unbiased with a small-sample correction. This correction is straightforward: it consists in multiplying the plain OLS estimator by a function of the sample size. In large samples, this function equals one, so the unbiased OLS estimator retains the same asymptotic properties as the plain OLS method, such as consistency and normality. In small samples, however, this function is strictly below one, correcting the positive bias of the plain OLS estimator. This result relates to \cite{gabaix_ibragimov2011rank}, which suggests a modification of the standard rank-size regression to deal with the OLS bias. However, while their procedure only mitigates the problem, ours completely eliminates it. Second, we demonstrate a close numerical relationship between ML and OLS estimators, despite their distinct theoretical foundations. The MLE of $d$ can be obtained from the rank-size regression by applying (i) a different small-sample adjustment and (ii) weighted least squares (WLS) instead of OLS. This result is equivalent to the WLS interpretation of the MLE provided in \cite{beirlant_etal1996tail}. Remarkably, the difference in small-sample corrections is typically small. Consequently, the main source of numerical difference between these OLS and ML estimators lies in the weighting scheme, with MLE assigning greater importance to smaller observations, particularly in large samples.
	
	Our first result addresses the small-sample bias of the OLS estimator, making it a viable alternative for estimating tail exponents. Nevertheless, the MLE remains both BLUE and MVUE. Therefore, if the distribution is strictly Pareto over the entire support, you should still use the MLE. However, the distribution may be just Pareto-like rather than strictly Pareto, being strictly Pareto only in the upper tail or regularly varying. And, in such cases, OLS methods may perform better, especially in large samples, because the MLE assigns greater weight to the smallest, non-Pareto, observations, as implied by our second theoretical result. Consequently, the MLE will be more biased, a finding that we illustrate through Monte Carlo simulations. Additionally, our second theoretical result indicates that the MLE should be less robust to the choice of estimation support. This was observed by \cite{aban_meerschaert2004generalized} in their empirical investigation of the daily trading volume of Amazon, Inc. stock. \cite{kratz_resnick1996qq} also found this in both empirical and Monte Carlo exercises. We confirm this finding through our own empirical exercises, studying (i) the distribution of US cities by population and (ii) the distribution of S\&P absolute daily returns. Hence, OLS provides more stable estimates across different supports, a desirable feature when the distribution is (approximately) Pareto only in the upper tail, as considering different supports is an effective way to deal with the uncertainty regarding the appropriate sample.
		
	All in all, the MLE should be applied when you are certain that the distribution under study is strictly Pareto. If the distribution could be just Pareto-like, OLS methods may be preferable. The issue is that, in practice, distinguishing between a strictly Pareto distribution and a Pareto-like distribution can be nearly impossible, making it difficult to determine which estimator is most appropriate. This is why you should \textit{also} use OLS estimation of tail exponents.
		
	\bigskip
	The remainder of the paper proceeds as follows. Section \ref{sec:formal_results} presents a formal statement of our results. These theoretical findings are substantiated through Monte Carlo simulations and real-world applications in Sections \ref{sec:mc} and \ref{sec:emp_exs}, respectively. Finally, Section \ref{sec:conclusion} concludes. 

	\section{Formal statement of the results} \label{sec:formal_results}
	
	Suppose a random variable $S$, $S\geq \underline{s}>0$, is Pareto distributed with a tail exponent $k>0$, and thus its survival function is given by \eqref{eq:sf_pareto}. Both $\underline{s}$ and $k$ are unknown. Let $S_1,S_2,...,S_n$ be a random sample of $S$ and $S_{(1)} \geq S_{(2)} \geq \dots \geq S_{(n)}$ be the associated (reversed) order statistics.\footnote{Following \cite{hill1975simple}, we reverse the usual definition of order statistics to simplify the formulas.}
		
	We analyze two estimators of $d \equiv 1/k$. First, the MLE derived in \cite{hill1975simple} as
	\begin{align}
		\hat{d}_{sML} = & \frac{\sum_{i=1}^{n-1} \ln(S_{(i)}/S_{(n)})}{n-1} . \label{eq:d_sMLE}
	\end{align}	
	We include a standard small-sample correction, dividing by $n-1$ instead of the original $n$, since this estimator is known to be unbiased and efficient, being both BLUE and MVUE (\citealt{aban_meerschaert2004generalized}). It is also consistent, with $\sqrt{n} \left(\hat{d}_{sML} - d\right)  \xrightarrow[]{d} \mathcal{N} \left(0,d^2\right)$ (\citealt{hall1982some}).
			
	Second, we apply OLS to the natural logarithm of \eqref{eq:sf_pareto}, replacing the lower bound $\underline{s}$ and the survival function $P(S\geq s)$ by the respective empirical counterparts $S_{(n)}$ and $\hat{P}(S\geq S_{(i)}) = i/n$.\footnote{The empirical probability $\hat{P}(S\geq S_{(i)})$ is the share of observations higher than or equal to $S_{(i)}$, which is exactly $i/n$ as $i$ is the observation's rank, that is, $S_{(i)}$ is the $i$-th largest observation in the sample.} Formally, the regression equation is
	\begin{align} 
		\ln \left(S_{(n)}/S_{(i)}\right) = & d\ln \hat{P}(S\geq S_{(i)}) + \epsilon_i = d\ln (i/n) + \epsilon_i \label{eq:reg_eq_d}
	\end{align}
	for $i=1,2,...,n-1$, where $\epsilon$ is an empirical error. As a result, the OLS estimator is 
	\begin{align} 
		\hat{d}_{OLS} = & \frac{\sum_{i=1}^{n-1} \ln (i/n) \ln(S_{(n)}/S_{(i)})}{\sum_{i=1}^{n-1} \left[\ln (i/n)\right]^2} . \label{eq:d_OLS}
	\end{align}	
	Note that we dropped the observation $i=n$ from the regression. However, this is done without loss of generality, as including such observation would not change the numerical result of \eqref{eq:d_OLS}. This estimator is known to be biased in small samples, but it is consistent, with $\sqrt{n} \left(\hat{d}_{OLS} - d\right)  \xrightarrow[]{d} \mathcal{N} \left(0,\frac{5}{4} d^2\right)$ (\citealt{schluter2018top}). This method consists in running a standard rank-size regression, but restricting the intercept value instead of freely estimating it.\footnote{Rewriting \eqref{eq:reg_eq_d}, $\ln S_{(n)} - \ln S_{(i)} = d\ln i - d\ln n + \epsilon_i \leftrightarrow \ln S_{(i)} = \left(\ln S_{(n)} + d\ln n\right) - d\ln i - \epsilon_i $. Thus, the intercept is constrained to $\ln S_{(n)} + d\ln n = \ln\left(S_{(n)}n^d\right)$.} Although such intercept estimation is common in the literature, it is not the most recommended practice. First, typically these regression lines are not valid distributions as $P(S\geq \underline{s}) \neq 1$ (\citealt{clauset_etal2009power}).\footnote{Similarly, \citet[p.254]{urzua2011testing} argues that ``the intercept is not a nuisance parameter in the regression.''} Second, this intercept constraint improves estimation efficiency, at least in large samples. Indeed, for the OLS estimator of $d$ \textit{with} intercept estimation, $\hat{d}_{OLS_c}$, it is well known that $\sqrt{n} \left(\hat{d}_{OLS_c} - d\right)  \xrightarrow[]{d} \mathcal{N} \left(0,2 d^2\right)$ (\citealt{kratz_resnick1996qq,csorgo_viharos1997asymptotic}), meaning its asymptotic variance is $2/(5/4)=1.6$ times higher than that of $\hat{d}_{OLS}$. Third, in small samples, $\hat{d}_{OLS_c}$ is more biased than $\hat{d}_{OLS}$, as shown in Appendix \ref{sec:app_formal_results_d_OLSc}.
	
	As \cite{aban_meerschaert2004generalized} point out, since $S$ is Pareto distributed with exponent $k$, $\ln (S/\underline{s})$ has an exponential distribution with rate $k = 1/d$.\footnote{To see that, just note $P\left[\ln (S/\underline{s}) \geq \ln (s/\underline{s})\right] = P(S \geq s) = (\underline{s}/s)^{k} = \exp[-k\ln(s/\underline{s})]$.} Consequently, from well known properties of the exponential distribution, (i) $\E\left[ \ln \left(S_i / \underline{s}\right) \right] = d$ and (ii) $\E\left[ \ln \left(S_{(i)} / \underline{s}\right) \right] = \sum_{j=i}^n (d/j)$. Using these moments, we can easily obtain our two main results.\footnote{From these moments, it is also straightforward to show that $\hat{d}_{ML} = \frac{n-1}{n} \hat{d}_{sML}$ is biased, but $\hat{d}_{sML}$ it is not. After all, given Equation \eqref{eq:d_sMLE}, $\E\left(\hat{d}_{sML}\right) = \frac{\sum_{i=1}^{n} \E\left[ \ln \left(S_i / \underline{s}\right)-\ln \left(S_{(n)} / \underline{s}\right) \right]}{n-1} = \frac{\sum_{i=1}^{n} (d-d/n)}{n-1} = \frac{n-1}{n-1} d = d$.} 
	
	The first proposition derives a simple way to correct the bias of the OLS estimator:
	
	\begin{proposition} \label{prop:d_OLS_bias}
		The shifted OLS estimator $\hat{d}_{sOLS} \equiv g(n) \hat{d}_{OLS}$ is unbiased for
		\begin{align}  
			g(n) \equiv \frac{\sum_{i=1}^{n-1} \left[\ln (i/n)\right]^2}{ n^{-1} \sum_{i=1}^{n-1} \ln (i/n) -\sum_{j=1}^{n} j^{-1} \sum_{i=1}^{j}  \ln (i/n)} . \notag
		\end{align}	
	\end{proposition}
	\begin{proof}
		From Equation \eqref{eq:d_OLS} and given that $\E\left[ \ln \left(S_{(i)} / \underline{s}\right) \right] = \sum_{j=i}^n (d/j)$,
		\begin{align} 
			\E\left(\hat{d}_{OLS}\right) = & \frac{\sum_{i=1}^{n-1} \ln (i/n) \E \left[\ln(S_{(n)}/\underline{s})-\ln(S_{(i)}/\underline{s})\right]}{\sum_{i=1}^{n-1} \left[\ln (i/n)\right]^2} = \frac{\sum_{i=1}^{n-1} \ln (i/n) \left[(d/n) - \sum_{j=i}^n (d/j) \right]}{\sum_{i=1}^{n-1} \left[\ln (i/n)\right]^2} \notag \\
			\E\left(\hat{d}_{OLS}\right) = & d \left\{\frac{n^{-1}\sum_{i=1}^{n-1} \ln (i/n)-\sum_{j=1}^{n} j^{-1} \sum_{i=1}^{j} \ln (i/n)}{\sum_{i=1}^{n-1} \left[\ln (i/n)\right]^2} \right\} = \frac{d}{g(n)} . \notag
		\end{align}	
		As a consequence, $\E \left( \hat{d}_{sOLS}\right) = g(n) \E \left( \hat{d}_{OLS}\right) = d$. 
	\end{proof}	
	\bigskip
	
	\cite{gabaix_ibragimov2011rank} suggest running OLS rank-size regressions but using the shifted rank $i-1/2$ instead of the rank $i$. They show that $1/2$ is the shift that optimally reduces the bias (up to leading order terms) in regressions \textit{with} intercept estimation. \cite{schluter2018top} finds similar results also (i) without such intercept estimation and (ii) when the distribution is regularly varying rather than strictly Pareto. While this procedure mitigates the problem, it does not solve it. As Proposition \ref{prop:d_OLS_bias} shows, to completely vanish the bias, one should multiply the OLS estimator \eqref{eq:d_OLS} by $g(n)$. We also derive a shift that corrects the bias for the OLS estimator of $d$ \textit{with} intercept estimation. This is discussed in Appendix \ref{sec:app_formal_results_d_OLSc}, Proposition \ref{prop:d_OLSc_bias}.
			
	In large samples, $g(n) = 1$, and consequently, the shifted OLS estimator retains the same asymptotic properties as the plain OLS method.\footnote{The bias of the plain OLS estimator $\hat{d}_{OLS}$ vanishes as the sample size increases. Consequently, $g(n) \to 1$ as $n\to \infty$, since otherwise, this would imply that $\hat{d}_{sOLS}$ is asymptotically biased, contradicting Proposition \ref{prop:d_OLS_bias}.} That is, $\sqrt{n} \left(\hat{d}_{sOLS} - d\right)  \xrightarrow[]{d} \mathcal{N} \left(0,\frac{5}{4} d^2\right)$. However, in small samples, this function is strictly below one (Figure \ref{fig:fn_gn}), correcting the positive bias of the plain OLS estimator. This makes the shifted OLS a viable alternative for the estimation of tail exponents. However, as the shifted MLE is both BLUE and MVUE, you should still rely on this estimator if the distribution is strictly Pareto over the entire support. We advocate using the OLS method because the distribution may be just Pareto-like rather than strictly Pareto. Our reasoning relies on our second proposition, which demonstrates a close numerical relationship between the unbiased estimators based on ML and OLS principles ($\hat{d}_{sOLS}$ and $\hat{d}_{sML}$, respectively):
			 
	\begin{proposition} \label{prop:d_ML_WLS}
		Denote by $\hat{d}_{WLS}$ the WLS estimator of $d$ on \eqref{eq:reg_eq_d} with weights $w_i = -1/\ln \hat{P}(S\geq S_{(i)}) = 1/\ln(n/i)$ for $i=1,...,n-1$. Then, $\hat{d}_{sML} = f(n) \hat{d}_{WLS}$ for 
		\begin{align}  
			f(n) \equiv - \frac{\sum_{i=1}^{n-1} \ln (i/n)}{n-1} . \notag
		\end{align}
	\end{proposition}
	\begin{proof}
		Applying WLS with weights $w_i = 1/\ln(n/i)$ for $i=1,...,n-1$ to \eqref{eq:reg_eq_d},
		\begin{align} 
			\hat{d}_{WLS} = & \frac{\sum_{i=1}^{n-1} w_i \ln (i/n)  \ln(S_{(n)}/S_{(i)})}{\sum_{i=1}^{n-1} w_i \left[\ln (i/n)\right]^2}  = \frac{(n-1)^{-1}\sum_{i=1}^{n-1} \ln(S_{(i)}/S_{(n)})}{-(n-1)^{-1}\sum_{i=1}^{n-1} \ln (i/n)} \notag \\
			\therefore \hat{d}_{sML} = & f(n) \hat{d}_{WLS} , \notag
		\end{align}
		where in the last line we use \eqref{eq:d_sMLE}.
	\end{proof}	
	\bigskip
	
	Proposition \ref{prop:d_ML_WLS} is equivalent to the WLS interpretation of the MLE provided in \cite{beirlant_etal1996tail}, demonstrating that, in large samples, the MLE functions as a WLS estimator that assigns greater weight to observations in the lower tail. Similar result applies for the dual problem of estimating $k=1/d$ (Proposition \ref{prop:k_ML_WLS} in Appendix \ref{sec:app_formal_results_k}). For OLS regressions with \cite{gabaix_ibragimov2011rank} correction, analogous results can also be obtained for the estimators of both $d$ and $k$ if we maintain the constrained intercept estimation.\footnote{For the estimator of $d$, just use $w_i = 1/\ln\left[n/(i-1/2)\right]$ for $i=1,...,n$ instead of the weights employed in Proposition \ref{prop:d_ML_WLS}. Regarding the estimator of $k$, one can still consider the same weights of Proposition \ref{prop:k_ML_WLS}.}
			
	In summary, from Propositions \ref{prop:d_OLS_bias} and \ref{prop:d_ML_WLS}, $\hat{d}_{sOLS} = g(n) \hat{d}_{OLS}$ and $\hat{d}_{sML} = f(n) \hat{d}_{WLS}$, meaning there are two sources of numerical differences between these two unbiased estimators. First, they employ different small-sample corrections: $g(n)$ for OLS and $f(n)$ for ML. However, these adjustments would be quite similar in practice, as illustrated in Figure \ref{fig:fn_gn}. Second, the ML method assigns different weights to each observation, giving more importance to smaller observations. This fact is shown in Figure \ref{fig:d_weights}, in which we normalize the weights to add up to one. We consider $n=3,10,100,10000$. As can be seen, the ML estimator assigns significantly greater weight to observations in the lower tail, particularly for larger samples. Indeed, for $n=10000$, the 20\% smallest observations have approximately 90\% of the weight!

	\begin{figure}[p]
		\centering
		\includegraphics[scale=1]{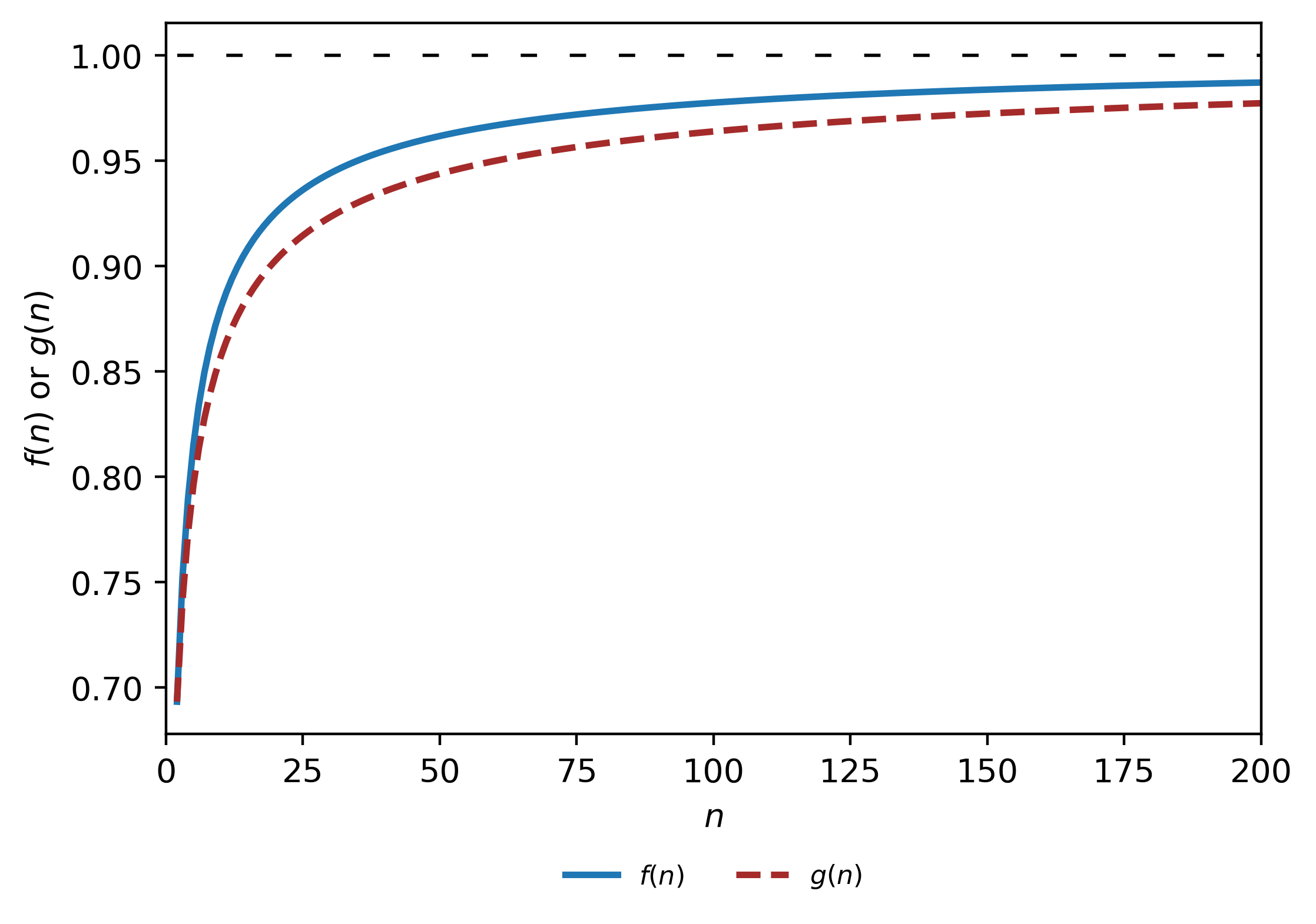}
		\caption{ML bias correction $f(n)$ \textit{versus} OLS bias correction $g(n)$}
		\label{fig:fn_gn}
	\end{figure}
	
	\begin{figure}[p]
		\centering
		\includegraphics[scale=1]{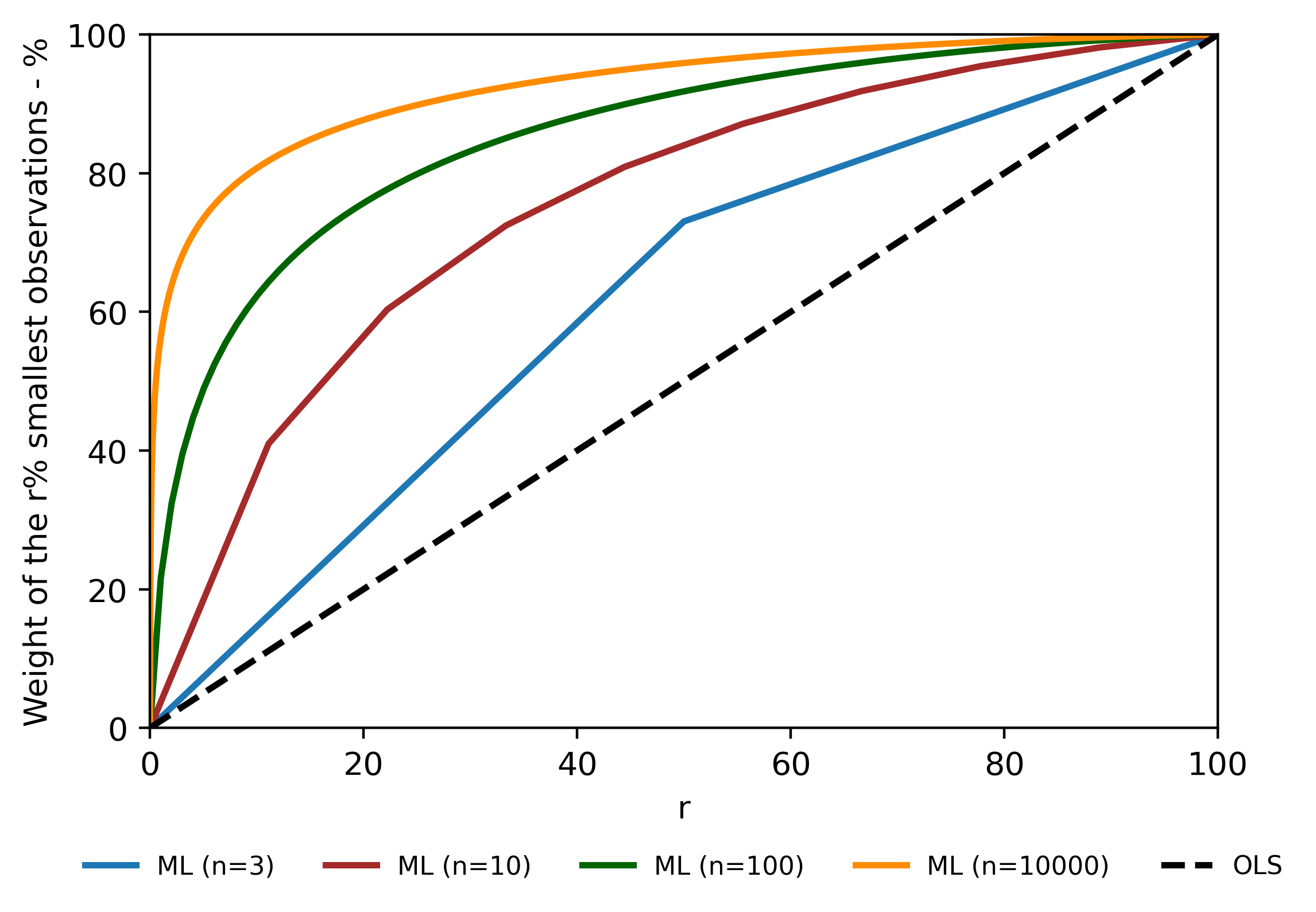}
		\caption{Observations' weights for $d$ estimators: ML \textit{versus} OLS}
		\label{fig:d_weights}
	\end{figure}
	
	Therefore, the most relevant source of numerical difference between $\hat{d}_{sOLS}$ and $\hat{d}_{sML}$ is the observations' weights. But why should we care about this numerical property? Does it mean MLE weights smaller observations excessively? If the distribution is strictly Pareto, when both (shifted) estimators are unbiased, the answer is no, since the MLE weighting is, in some sense, optimal, as the MLE is both BLUE and MVUE. However, if the distribution is just Pareto-like rather than strictly Pareto, these MLE weights may not be ideal. Indeed, the OLS estimator may perform better under two commonly studied deviations from a strictly Pareto distribution: (i) strictly Pareto but only in the upper tail, and (ii) regularly varying. In these cases, a Pareto distribution holds (approximately, in the second case) only in the upper tail, above some unknown threshold. There are several methods in the literature to estimate this threshold (e.g., \citealt{breiman_etal1990robust}, \citealt{dekkers_dehaan1993optimal}, \citealt{drees_kaufmann1998selecting}, \citealt{danielsson_etal2001using}, \citealt{handcock_jones2004likelihood}, and \citealt{clauset_etal2007frequency}), but there is always some uncertainty, and you may end up with some smaller observations that are not Pareto in your sample. In such scenario, one should expect a higher bias for the MLE since it will assign greater weight to these smallest, non-Pareto, observations. In Section \ref{sec:mc}, we confirm this through Monte Carlo experiments.
			
	Moreover, in such cases where the distribution is (roughly) Pareto only in the upper tail, it may be recommended to estimate $d$ for different supports, given the uncertainty regarding the appropriate sample. In this scenario, you may want an estimator that is more robust to these support choices, which should be the OLS estimator. After all, when we narrow the support, dropping the smallest observations from the sample, the ML estimate is expected to change more significantly, given the greater weights it assigns to these smaller observations. Investigating the daily trading volume of Amazon, Inc. stock from March 20, 1995, to June 18, 2001, \cite{aban_meerschaert2004generalized} find empirical results that are consistent with this reasoning. The increased robustness of the OLS method is also found by \cite{kratz_resnick1996qq} in both empirical and Monte Carlo exercises.\footnote{Their Monte Carlo experiment uses a Pareto with $k=1$, while they empirically assess ``[...] interarrival times between packets generated and sent to a host by a terminal during	a logged-on session'' (\citealt[p.720]{kratz_resnick1996qq}).} We also observe this result in our own empirical exercises, discussed in Section \ref{sec:emp_exs}.
			
	In practice, it may be virtually impossible to distinguish a strictly Pareto distribution from a Pareto-like distribution. And if your distribution is just Pareto-like, you should address the issue of choosing the ideal estimation support. Given the uncertainty involved, it is beneficial to employ an estimator that (i) performs better when non-Pareto observations are (incorrectly) included in the sample and (ii) is more robust to the choice of estimation sample. As we just argued, the OLS method surpasses the MLE in both respects. This is why we advocate for also using the OLS estimator in real-world applications. 
	
	\section{Monte Carlo exercises} \label{sec:mc}
		
	In the previous section, we argued, based on theoretical results, that the OLS estimator may outperform the MLE if the distribution is (i) strictly Pareto but only in the upper tail, or (ii) regularly varying. In this section, we investigate this through Monte Carlo simulations, considering these two Pareto-like distributions. We evaluate four estimators of $d$: (i) the shifted MLE $\hat{d}_{sML}$, (ii) the OLS estimator $\hat{d}_{OLS}$, (iii) the shifted OLS estimator $\hat{d}_{sOLS} = g(n) \hat{d}_{OLS}$, and (iv) the OLS estimator with \cite{gabaix_ibragimov2011rank} correction $\hat{d}_{OLS-GI}$, in which we run the OLS regression \eqref{eq:reg_eq_d} again, but replacing $i$ with $i-1/2$. We assess their performance through estimators' bias, variance, and Mean Squared Error (MSE), each computed over 10,000 simulations. 
					
	\subsection{Strictly Pareto but only in the upper tail} \label{sec:mc1}
	
	As in \cite{clauset_etal2009power}, we suppose the distribution of $S$, $S \geq 0$, follows a Pareto at $\underline{s}$ and above but an exponential below, using the density function
	\begin{align} \label{eq:density_mc1}	
		f(s) = &
		\begin{cases}
			C e^{(k+1)(1-s/\underline{s})} & \text{, if } 0 \leq s < \underline{s} \\
			C (\underline{s}/s)^{k+1} & \text{, if } s \geq \underline{s}
		\end{cases}
	\end{align}	
	where $k,\underline{s}>0$, and $C>0$ is a normalizing constant that ensures the probabilities add up to one over the support $[0,+\infty)$.\footnote{From \eqref{eq:density_mc1}, $P(S < \underline{s}) =  C \int_0^{\underline{s}} e^{(k+1)(1-s/\underline{s})} ds = \frac{C \underline{s}}{k+1} \left( e^{k+1} - 1  \right)$ and $P(S \geq \underline{s}) =  C \underline{s}^{k+1} \int_{\underline{s}}^{\infty} s^{-(k+1)} ds =  \frac{C \underline{s}}{k}$. Therefore, $1 = P(S < \underline{s}) + P(S \geq \underline{s}) = \frac{C \underline{s}}{k+1} \left( e^{k+1} - 1  \right) + \frac{C \underline{s}}{k} \to C = \frac{k(k+1)}{\underline{s}(1+ke^{k+1})}$.} We set $k=1$ and $\underline{s}=50$, yielding the survival function shown in Figure \ref{fig:sf_mc1}, which has the expected power law shape for $s\geq50$. To evaluate the estimators' performance also when the Pareto support is not properly identified, we sample over the non-Pareto support $S\geq5$ in addition to the ideal $S\geq 50$. For each support, we draw different numbers of observations between 3 and 100 from this distribution. 
	\begin{figure}[h!]
		\centering
		\includegraphics[scale=1]{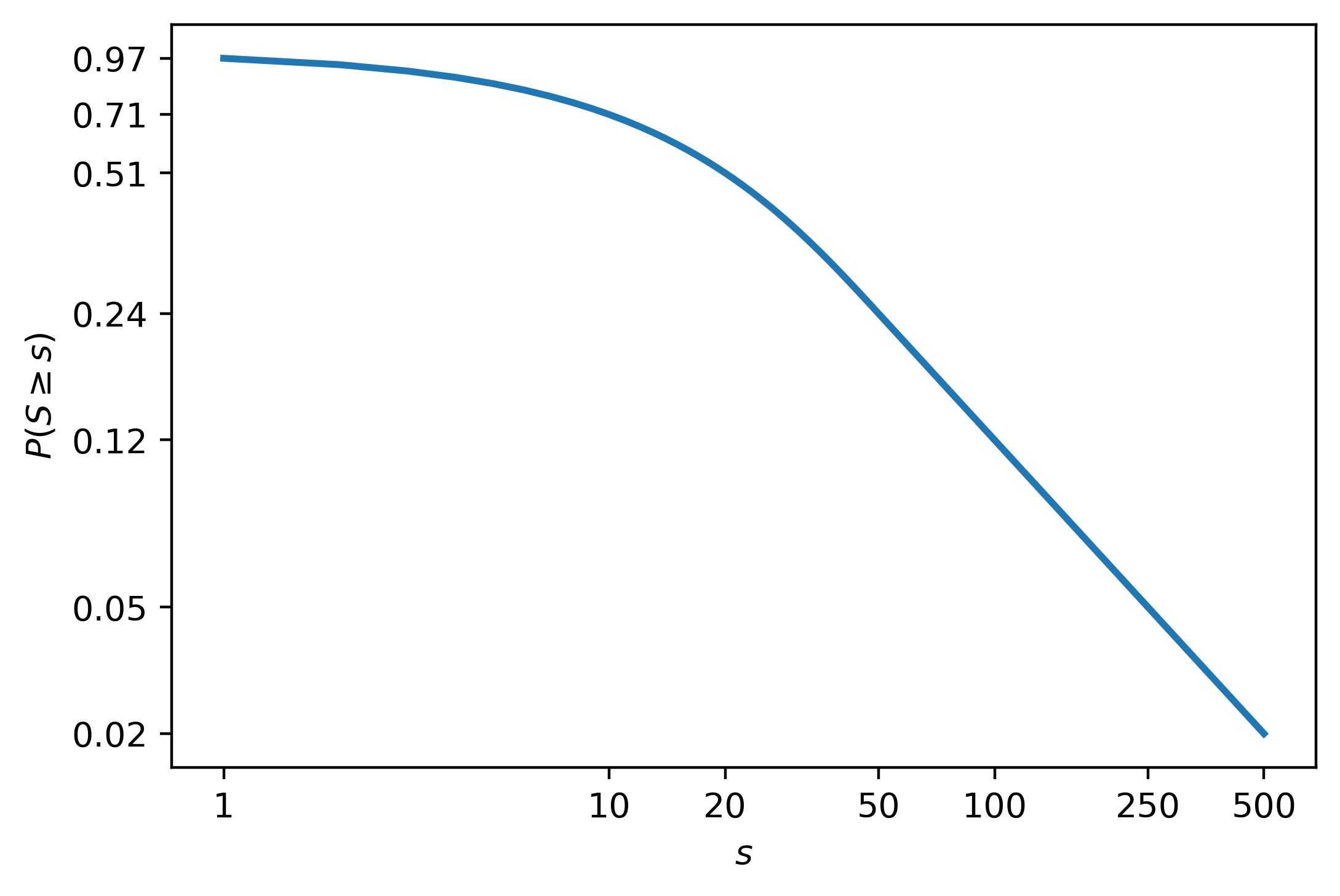}
		\caption{Survival function of the first Monte Carlo exercise (axes in logarithmic scale)}
		\label{fig:sf_mc1}
	\end{figure}
			
	Figure \ref{fig:mc1_d} plots the performance measures for the estimators of $d$, for each estimation support and number of observations. Under the ideal support $S\geq50$, all four estimators are consistent. The OLS estimator is biased and has the highest variance in small samples. The \cite{gabaix_ibragimov2011rank} method is also biased, but only in very small samples (say, less than 20 observations). The shifted ML and shifted OLS estimators are not biased, as expected given the results of Section \ref{sec:formal_results}. In terms of MSE, the shifted ML, shifted OLS, and \cite{gabaix_ibragimov2011rank} estimators show similar performance. 
		
	When we incorrectly use $S\geq5$, all estimators are biased, even in large samples. Consistent with our reasoning in Section \ref{sec:formal_results}, the shifted MLE is more biased and consequently has a higher MSE than the shifted OLS estimator, reflecting the greater weights the MLE assigns to the smallest, non-Pareto, observations. Interestingly, the relative performance of the shifted MLE becomes even worse as the sample gets larger, which is not surprising since the weights of the smallest observations increase with the sample size (Figure \ref{fig:d_weights}). Among all four evaluated estimators, the shifted MLE also has the worst performance in terms of both bias and MSE, except perhaps in very small samples. This result suggests that practitioners should also consider using OLS estimators, as there is always a risk of using an incorrect support in the estimation.
		
	Before moving on to the regularly varying case, it is worth noting that while we have only discussed the results for $k=1$, our main conclusions hold more generally. In Appendix \ref{sec:app_mc1}, we present the performance of the estimators for $k = 1/2,1,2,3$, all of which show qualitatively similar results.	
	
	\begin{figure}[h!]
		\centering
		\includegraphics[scale=1]{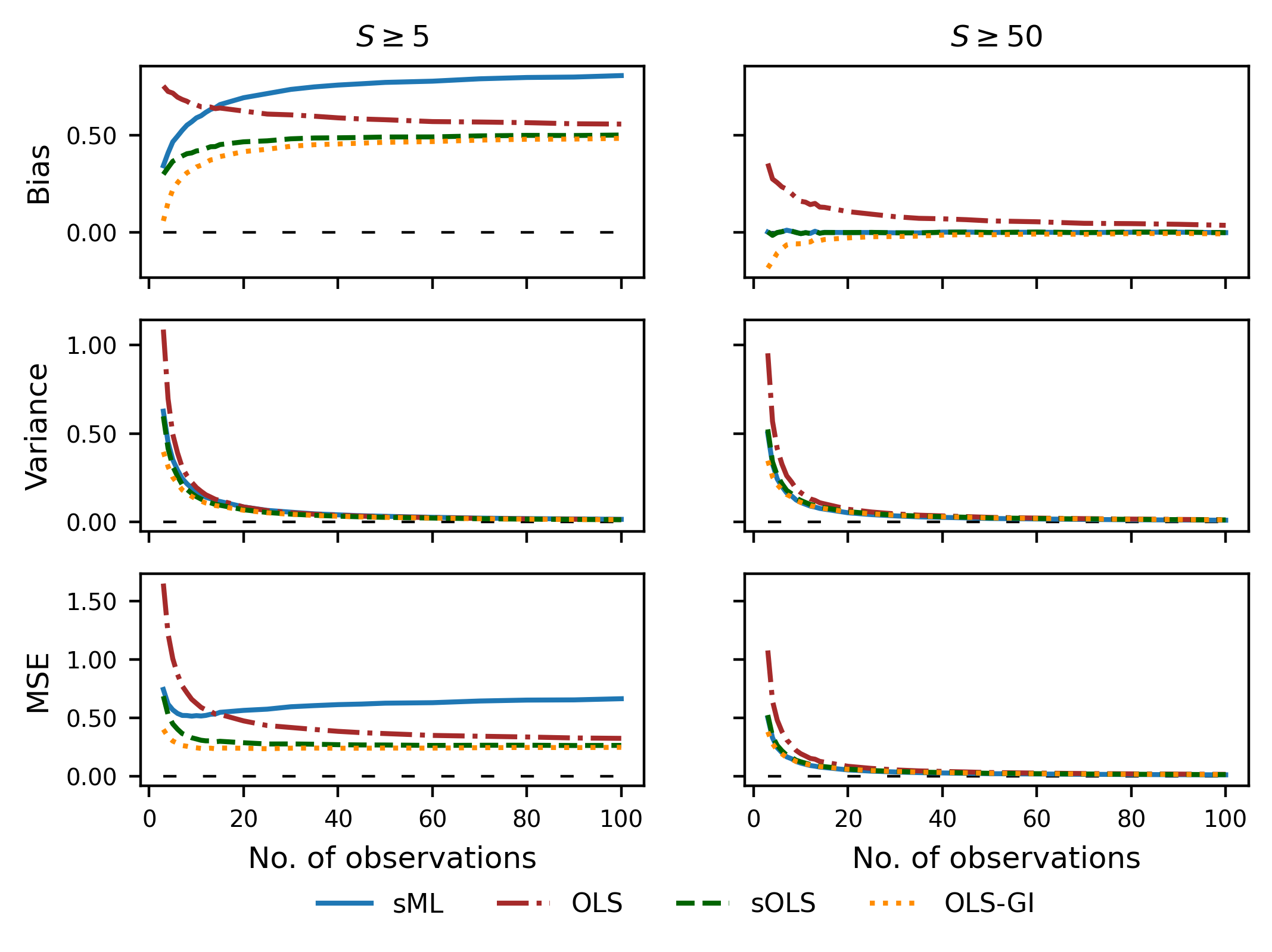}
		\caption{Performance of $d$ estimators in the first Monte Carlo exercise}
		\label{fig:mc1_d}
	\end{figure}

	\subsection{Regularly varying} \label{sec:mc2}
			
	Now, we study the performance of the estimators when the distribution is regularly varying, meaning $P(S \geq s) = s^{-k} L(s)$, where $L(s)$ is a slowly varying function, i.e., $L(ts)/L(s) \to 1$ as $s \to \infty$ for any $t>0$. More specifically, we assume $S$, $S \geq 0$, is Burr distributed, implying
	\begin{align} \label{eq:sf_mc2}	
		P(S \geq s) = & (1+s^{-k\rho})^{1/\rho} 
	\end{align}	
	for $s \geq 0$, $k>0$, and $\rho<0$. It is easy to see that this distribution varies regularly as $L(s) = s^k (1+s^{-k\rho})^{1/\rho} = (s^{k\rho} +1)^{1/\rho}$ is slowly varying.\footnote{For any $t>0$, $\lim_{s\to \infty} \frac{L(ts)}{L(s)} = \left(\frac{(st)^{k\rho} +1}{s^{k\rho} +1}\right)^{1/\rho}=1$.} Thus, this distribution approximates a Pareto with exponent $k$ in the upper tail as $\frac{P(S \geq s)}{P(S \geq t s)} \to t^k$ as $s\to \infty$ for any $t>0$, which is exactly the decay of the Pareto survival function \eqref{eq:sf_pareto}.\footnote{From \eqref{eq:sf_mc2}, $\lim_{s\to\infty} \frac{P(S \geq s)}{P(S \geq t s)} = \left[ \lim_{s\to\infty} \frac{1+s^{-k\rho}}{1+(t s)^{-k\rho}}\right]^{1/\rho} = \left( t ^{k\rho}\right)^{1/\rho} = t ^k$, while \eqref{eq:sf_pareto} yields $\frac{P(S \geq s)}{P(S \geq t s)} = t^k$.} Moreover, this approximation improves progressively as $\rho$ increases in magnitude. In particular, if $\rho \to -\infty$, the Burr density coincides with a Pareto distribution if defined over any strictly positive support.\footnote{If $S$, $S \geq \underline{s}>0$, is truncated Burr distributed, $P(S \geq s) = \frac{(1+s^{-k\rho})^{1/\rho}}{(1+\underline{s}^{-k\rho})^{1/\rho}}$. Note $\lim_{\rho \to -\infty} (1+s^{-k\rho})^{1/\rho} = \exp \left[ \lim_{\rho \to -\infty} \frac{\ln (1+s^{-k\rho})}{\rho} \right] = \exp \left[ \lim_{\rho \to -\infty} \frac{-k s^{-k\rho} \ln s }{1+s^{-k\rho}} \right] = \exp \left[ \ln \left(s^{-k}\right) \right] = s^{-k}$, with the second equality derived from the L'Hôpital's rule. As a result, $\lim_{\rho \to -\infty} P(S \geq s) = \frac{\lim_{\rho \to -\infty} (1+s^{-k\rho})^{1/\rho}}{\lim_{\rho \to -\infty} (1+\underline{s}^{-k\rho})^{1/\rho}} = (\underline{s}/s)^k$.} Given that, we assess the estimators' performance considering different values of $|\rho|$ between 1/2 and 5. In all cases, we set $k=1$. Figure \ref{fig:sf_mc2} plots the respective survival functions for $\rho=-5,-1,-1/2$. Consistent with the theoretical results, they show a Pareto shape in the upper tail, particularly for high $|\rho|$.
		
	\begin{figure}[h!]
		\centering
		\includegraphics[scale=1]{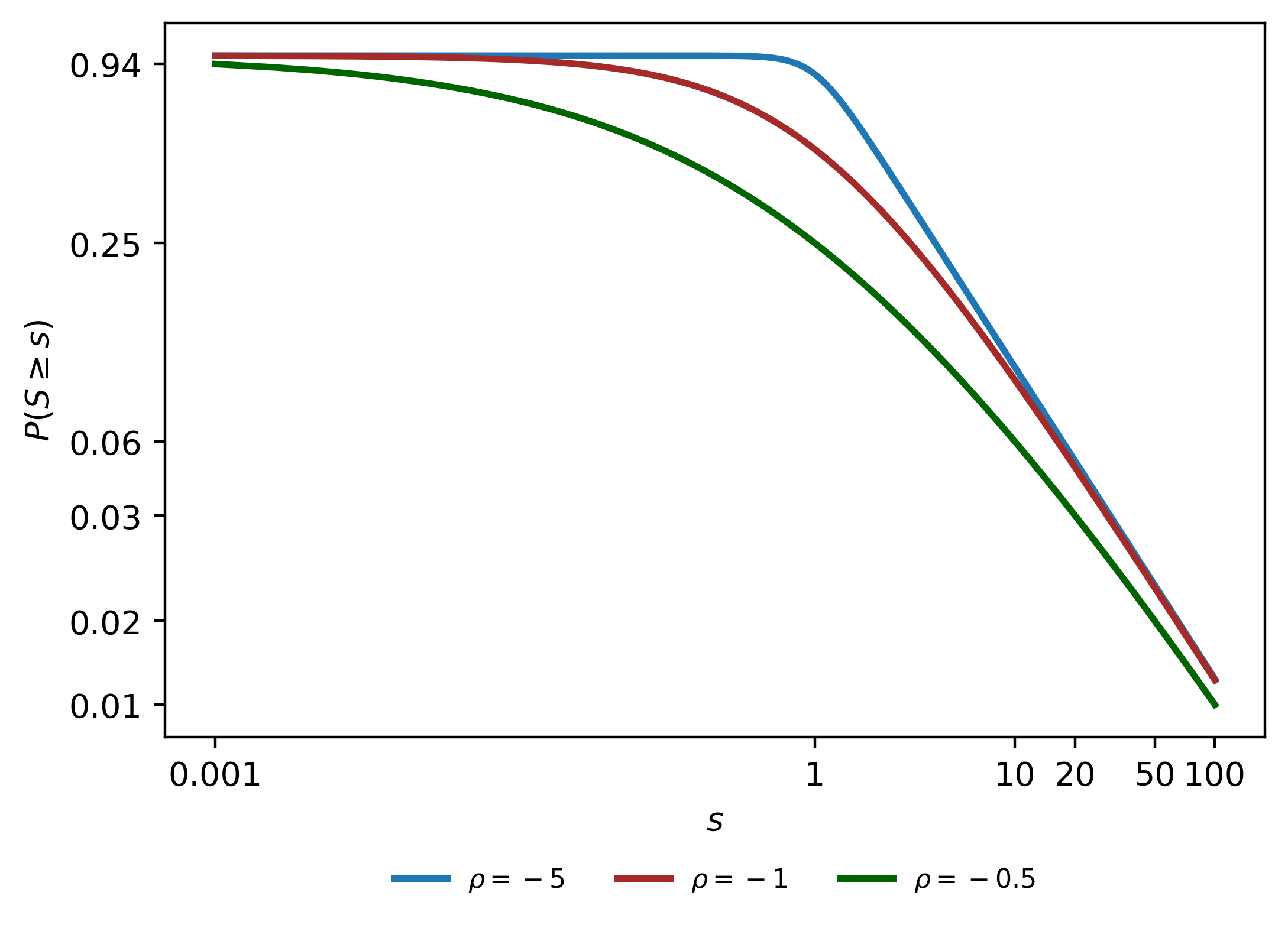}
		\caption{Survival function of the second Monte Carlo exercise (axes in logarithmic scale)}
		\label{fig:sf_mc2}
	\end{figure}

	Figure \ref{fig:mc2_d} plots the performance measures of the estimators for each $\rho$, considering either 10 or 100 observations, sampled over the entire non-negative support. As expected, the estimators are unbiased if the magnitude of $\rho$ is high, but they become progressively biased as $|\rho|$ decreases. Consistent with the results of Section \ref{sec:formal_results}, this bias is more significant for the shifted MLE, which also shows more variance and thus higher MSE. As in the exercise of Section \ref{sec:mc1}, the relative performance of the shifted MLE deteriorates further in larger samples, consistent with the results of Figure \ref{fig:d_weights}. Therefore, these results reaffirm the superior performance of OLS methods under deviations from a strictly Pareto distribution. In particular, for small samples, one should consider using the shifted OLS or \cite{gabaix_ibragimov2011rank} estimator, as they show less bias, variance, and MSE than the plain OLS method under such circumstances.
		
	As in the first Monte Carlo exercise, we have only discussed the results for $k=1$, but our main conclusions hold more generally. This can be seen in Appendix \ref{sec:app_mc2}, which presents the estimators' performance for $k = 1/2,1,2,3$.

	\begin{figure}[h!]
		\centering
		\includegraphics[scale=1]{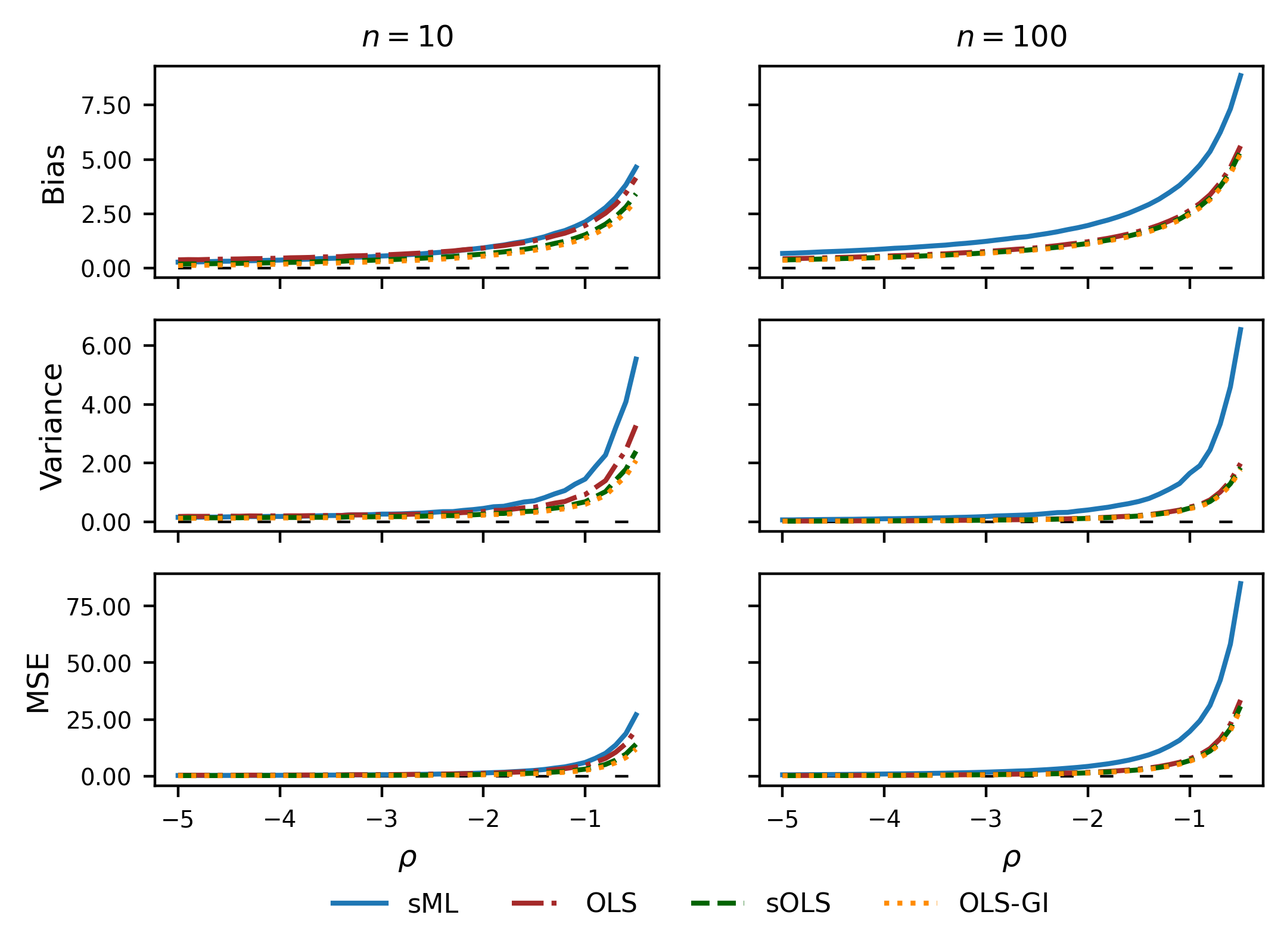}
		\caption{Performance of $d$ estimators in the second Monte Carlo exercise}
		\label{fig:mc2_d}
	\end{figure}

	\section{Empirical applications} \label{sec:emp_exs}
	
	In Section \ref{sec:formal_results}, we argued that OLS estimators should be more robust to the choice of the estimation support, which is a desirable feature as it may be recommended to estimate $d$ under different assumptions to deal with the uncertainty regarding the appropriate sample. This reasoning relied on theoretical findings, which we now substantiate with two real-world applications from \cite{gabaix2009power,gabaix2016power}. First, we study the distribution of US cities by population, using data for 366 Metropolitan Statistical Areas in 2010 from the Statistical Abstracts of the United States, published by the US Census Bureau. Second, we investigate the distribution of S\&P absolute daily returns between 04/28/2014 and 04/25/2024, totaling 2,515 observations that are taken from the Federal Reserve Economic Data (FRED). In both exercises, we estimate $d$ for all possible supports, employing the same four estimators of Section \ref{sec:mc}: (i) the shifted MLE $\hat{d}_{sML}$, (ii) the OLS estimator $\hat{d}_{OLS}$, (iii) the shifted OLS estimator $\hat{d}_{sOLS}$, and (iv) the OLS estimator with \cite{gabaix_ibragimov2011rank} correction $\hat{d}_{OLS-GI}$.
	
	\subsection{Distribution of US cities by population} \label{sec:emp_exs_UScities}

	Our first empirical assessment begins with Figure \ref{fig:SFemp_UScities_shade}, which plots the empirical survival function for US cities' populations. It shows a Pareto shape, at least when we disregard smaller cities and the inevitable noisy behavior in the empirical upper tail. The estimates of $d$ for all possible supports are shown in Figure \ref{fig:d_UScities_all} in Appendix \ref{sec:app_emp_exs_UScities}. However, the estimates vary so much across these supports that it is difficult to discern the differences between the estimators. Given that, we focus here only on the support's lower bounds depicted in the shaded area of Figure \ref{fig:SFemp_UScities_shade}. The estimates for these supports are illustrated in Figure \ref{fig:d_UScities_selec}. 
	
	As can be seen, the results vary with the estimation support, but the shifted MLE line is much more noisy, indicating that its estimates are more sensitive to the sample choice. This is consistent with our reasoning in Section \ref{sec:formal_results}. Furthermore, the shifted OLS and the OLS with \cite{gabaix_ibragimov2011rank} correction show very similar results across all these supports. In contrast, the plain OLS method produces higher estimates than the other two OLS-based methods, particularly in smaller samples, reflecting their bias adjustments. Additionally, note that all estimators' results decrease as the support narrows, with ML estimates consistently higher than OLS results, which may suggest that the distribution is regularly varying rather than strictly Pareto. Finally, it is also worth mentioning that the point estimates of $d$ are consistently around one, suggesting $k = 1/d \approx 1$. This result aligns with the literature (e.g., \citealt{gabaix1999zipf}, \citealt{gabaix_ioannides2004evolution}, \citealt{gabaix2009power}, and \citealt{gabaix2016power}) and is known as Zipf's law, named after the linguist George Kingsley Zipf, who observed similar empirical regularity in the frequency of word usage across different languages and countries (\citealt{zipf1949human}).
		
	\begin{figure}[p]
		\centering
		\includegraphics[scale=1]{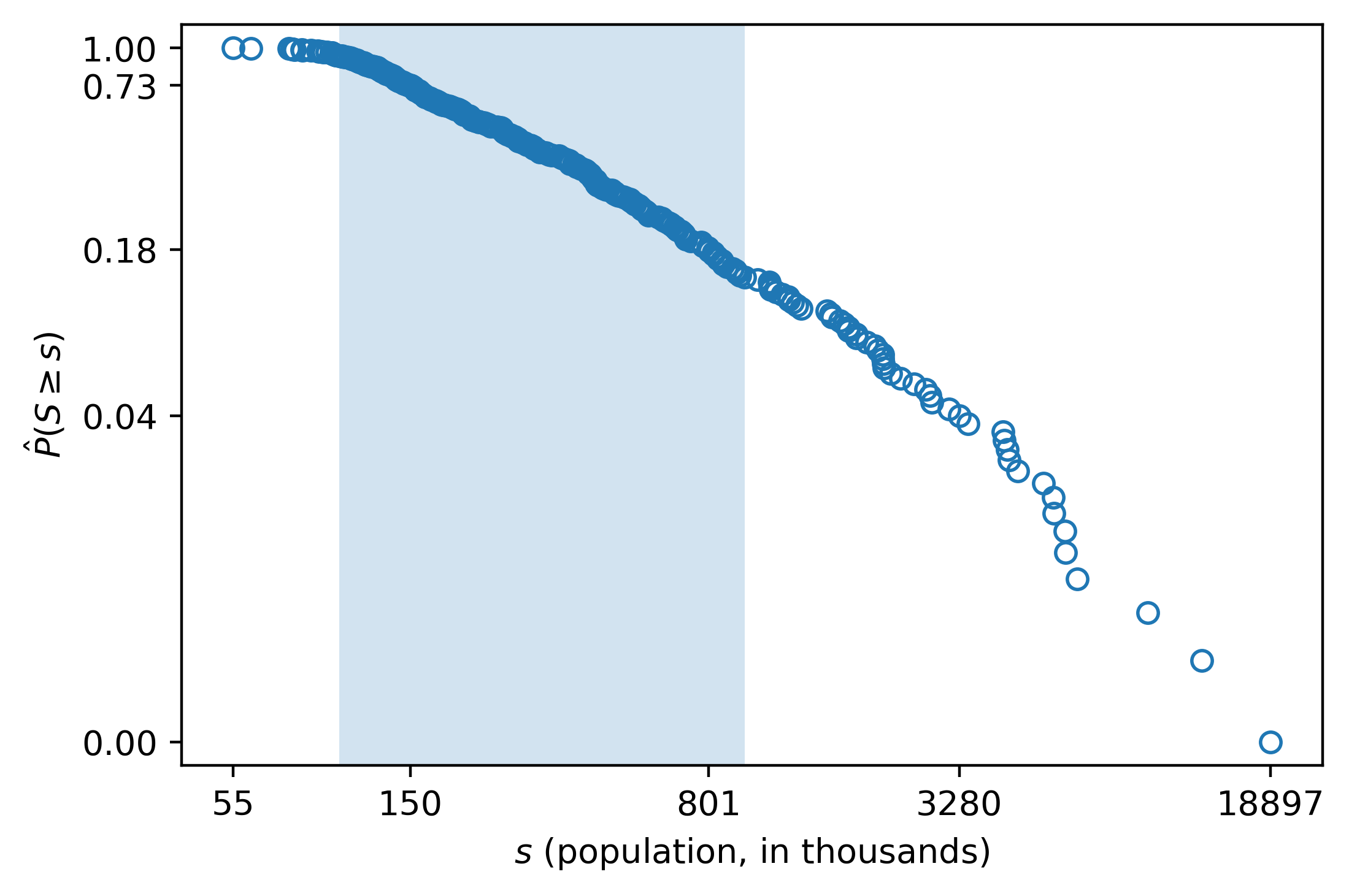}
		\caption{Empirical survival function for US cities' populations (axes in logarithmic scale)}
		\label{fig:SFemp_UScities_shade}
	\end{figure}

	\begin{figure}[p]
		\centering
		\includegraphics[scale=1]{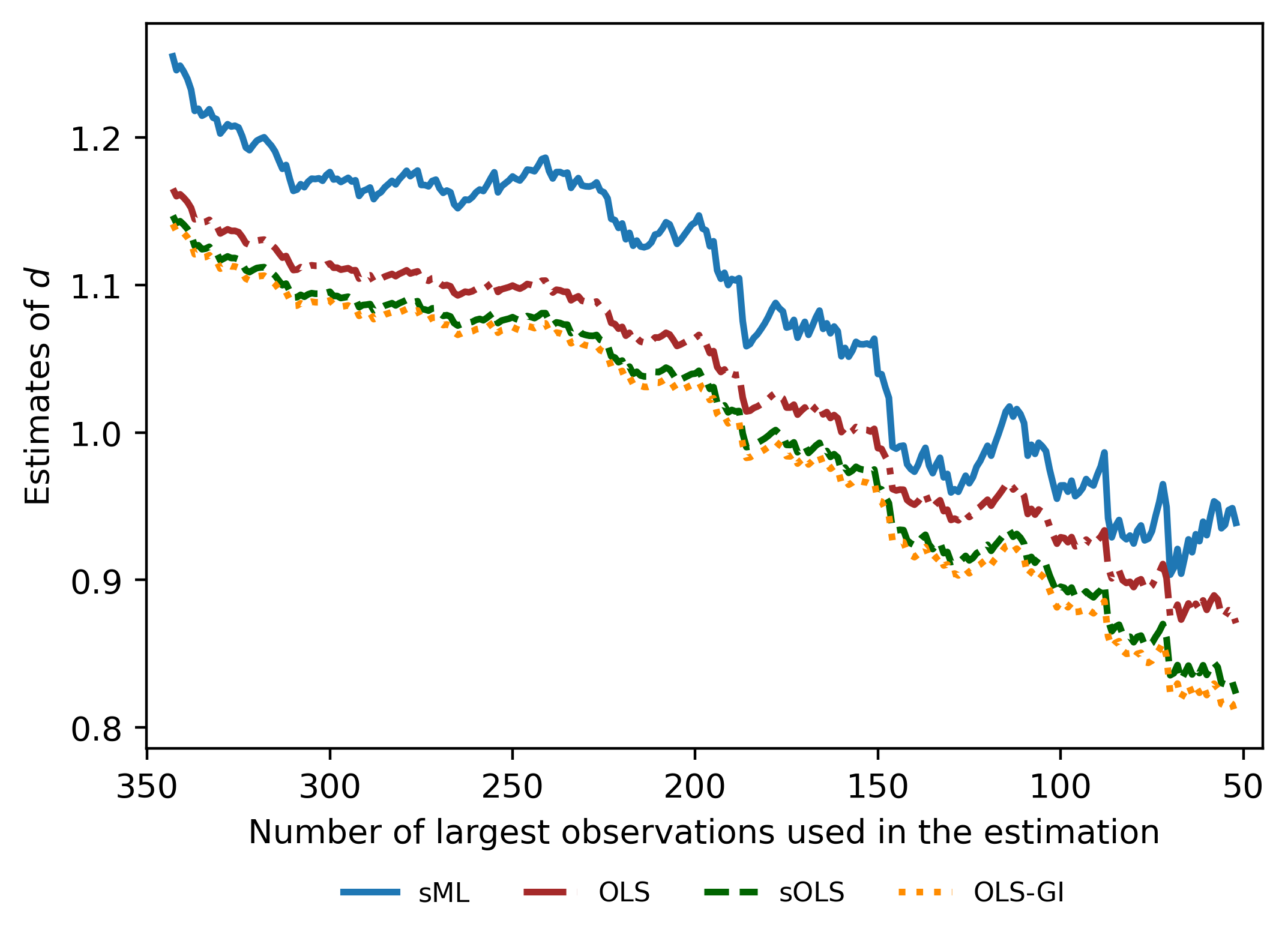}
		\caption{Estimates of $d$ for US cities' populations, selected supports}
		\label{fig:d_UScities_selec}
	\end{figure}
	
	\subsection{Distribution of S\&P absolute daily returns} \label{sec:emp_exs_S&Pr}
	
	We now move on to our second empirical exercise, evaluating S\&P absolute daily returns. Figure \ref{fig:SFemp_SPr_shade} presents the empirical survival function, which again exhibits a Pareto shape if we ignore low absolute returns and the noisy empirical behavior in the upper tail. As in the previous exercise, we focus on selected support's lower bounds, as identified in the shaded area of Figure \ref{fig:SFemp_SPr_shade}. Figure \ref{fig:d_SPr_selec} shows the corresponding estimates of $d$.\footnote{For results across all possible supports, please refer to Appendix \ref{sec:app_emp_exs_S&Pr}, Figure \ref{fig:d_SPr_all}.} Some of the conclusions from the previous exercise still apply. Firstly, the results vary with the estimation support, but this variation is particularly pronounced for the MLE, as its line is much more noisy. Secondly, the estimates of the shifted OLS and the OLS with \cite{gabaix_ibragimov2011rank} correction are very similar, while the plain OLS estimates are consistently higher, especially in smaller samples. 
	
	However, some other conclusions no longer hold. First, even though the estimates show an overall downward trend, this decrease starts to level off near the smallest sample sizes considered in Figure \ref{fig:d_SPr_selec}, when the shifted ML, shifted OLS, and OLS with \cite{gabaix_ibragimov2011rank} correction yield practically identical results. This may suggest that the distribution is strictly Pareto above some threshold. Second, the point estimates of $d$ are not close to one; instead, they are around 0.4, reaching levels near 1/3 for the narrowest supports evaluated here. This indicates $k = 1/d \approx 3$, which is known as the cubic law of stock price fluctuations. This law appears to apply to the tail distribution of short-term (15 seconds to a few days) returns, for positive and negative returns separately, across different stock sizes and countries (see, e.g., \citealt{gopikrishnan_etal1999scaling}, \citealt{gabaix2009power}, and \citealt{gabaix2016power}). 

	\begin{figure}[p]
		\centering
		\includegraphics[scale=1]{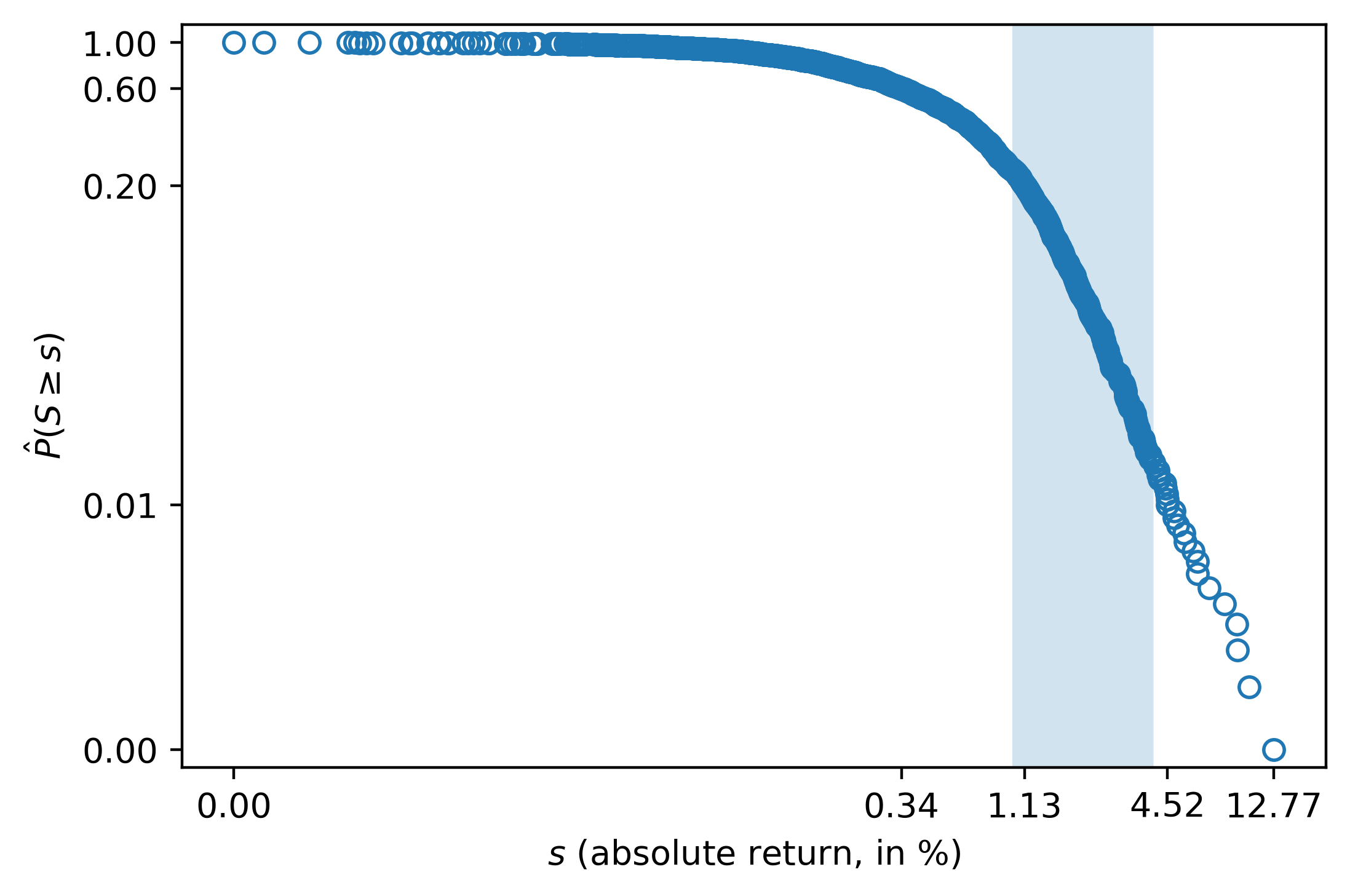}
		\caption{Empirical survival function for S\&P absolute daily returns (axes in logarithmic scale)}
		\label{fig:SFemp_SPr_shade}
	\end{figure}

	\begin{figure}[p]
		\centering
		\includegraphics[scale=1]{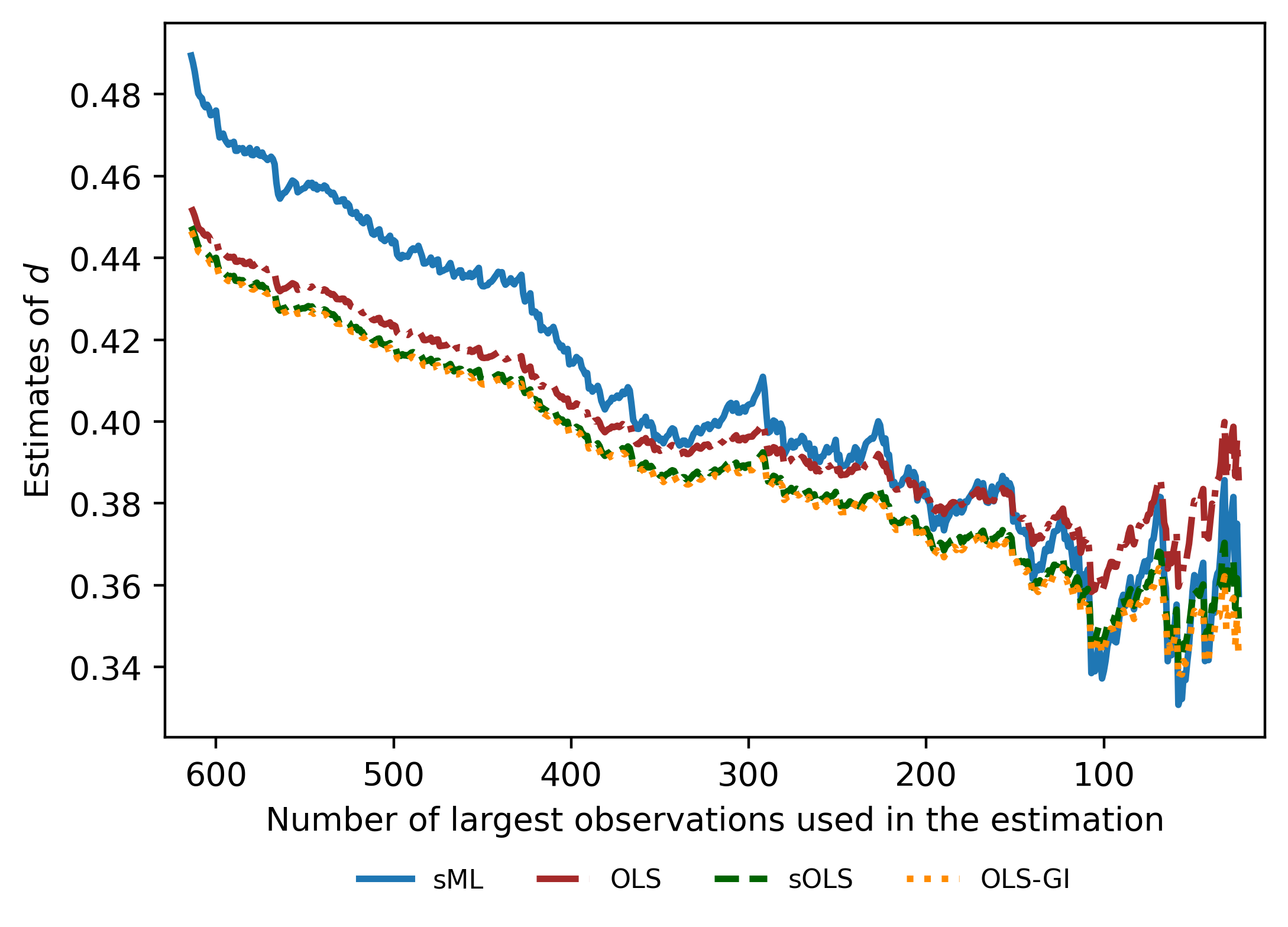}
		\caption{Estimates of $d$ for S\&P absolute daily returns, selected supports}
		\label{fig:d_SPr_selec}
	\end{figure}
	
	\section{Conclusion} \label{sec:conclusion}
		
	In this paper, we demonstrated two important theoretical results for the estimation of tail exponents. Firstly, we established that the estimator of $d$ from the OLS rank-size regression is also unbiased with a small-sample correction. Secondly, we proved that there is a close numerical relationship between ML and OLS methods. From this result, we showed that these estimators differ numerically mainly because MLE assigns greater weight to smaller observations, particularly in large samples.
		
	We argued that the shifted MLE should be used when the distribution under study is strictly Pareto. Although our shifted OLS estimator is unbiased according to our first theoretical result, the MLE method is both BLUE and MVUE. However, if the distribution is strictly Pareto only in the upper tail or regularly varying, OLS methods may be preferable, especially in large samples. This recommendation is based on our second theoretical result. Specifically, in these Pareto-like situations, the MLE has higher bias because it assigns greater weight to the smallest, non-Pareto, observations, as illustrated by our Monte Carlo simulations. Additionally, the MLE is less robust to the choice of estimation support, a fact observed in our real-world applications. Thus, OLS provides more stable estimates across different supports, which is a valuable feature when dealing with the uncertainty of the appropriate sample.
			
	In practice, distinguishing between a strictly Pareto distribution and a Pareto-like distribution can be nearly impossible, making it difficult to determine which estimator is most appropriate. This is why you should \textit{also} use OLS estimation of tail exponents.

\section*{Acknowledgements}

DOC thanks to CNPq for partial financial support (304706/2023-0).

\section*{Disclaimer}

The views expressed in this paper are those of the authors and should not be interpreted as representing the positions of the Banco Central do Brasil or its board members.

	\clearpage
	\begin{appendices}
		\appendix
		\counterwithin{figure}{section}
		\counterwithin{equation}{section}
		\counterwithin{table}{section}
		\counterwithin{proposition}{section}
		
		\section{Additional theoretical results} \label{sec:app_formal_results}
		
		\subsection{OLS with intercept estimation} \label{sec:app_formal_results_d_OLSc}
		
		Consider the regression	equation
		\begin{align} 
			\ln \left(S_{(n)}/S_{(i)}\right) = & \alpha + d\ln (i/n) + \varepsilon_i \label{eq:reg_eq_d_c}
		\end{align}
		for $i=1,2,...,n$, where $\varepsilon$ is an empirical error and both $\alpha$ and $d$ are unknown parameters. Denote by $\hat{d}_{OLS_c}$ the OLS estimator of $d$ in \eqref{eq:reg_eq_d_c}, that is,
		\begin{align} 
			\hat{d}_{OLS_c} = & \frac{n^{-1} \sum_{i=1}^n \ln (i/n) \ln(S_{(n)}/S_{(i)})-\left[n^{-1} \sum_{i=1}^n  \ln (i/n) \right]\left[n^{-1} \sum_{i=1}^n \ln(S_{(n)}/S_{(i)})\right]}{n^{-1} \sum_{i=1}^n \left[\ln (i/n)\right]^2-\left[ n^{-1} \sum_{i=1}^n \ln (i/n) \right]^2} \notag \\
			\hat{d}_{OLS_c} = & \frac{-\sum_{i=1}^n \ln (i/n) \ln(S_{(i)}/S_{(n)})+\left[\sum_{i=1}^n  \ln (i/n) \right]\left[n^{-1} \sum_{i=1}^n \ln(S_i/S_{(n)})\right]}{\sum_{i=1}^n \left[\ln (i/n)\right]^2-n^{-1} \left[\sum_{i=1}^n \ln (i/n) \right]^2} \notag \\
			\hat{d}_{OLS_c} = & \frac{\left[\sum_{i=1}^n  \ln (i/n) \right]\left[n^{-1} \sum_{i=1}^n \ln(S_i/\underline{s})\right]-\sum_{i=1}^n \ln (i/n) \ln(S_{(i)}/\underline{s})}{\sum_{i=1}^n \left[\ln (i/n)\right]^2-n^{-1} \left[\sum_{i=1}^n \ln (i/n) \right]^2} , \label{eq:d_OLSc}
		\end{align}	
		where in the last line we use $-\sum_{i=1}^n \ln (i/n) \ln(\underline{s}/S_{(n)})+\left[\sum_{i=1}^n  \ln (i/n) \right]\left[n^{-1} \sum_{i=1}^n \ln(\underline{s}/S_{(n)})\right] = \left[\sum_{i=1}^n  \ln (i/n) \right] \left[-\ln(\underline{s}/S_{(n)})+n^{-1}n\ln(\underline{s}/S_{(n)})\right] = 0$.
		
		Given this estimator, the following result holds:
		
		\begin{proposition} \label{prop:d_OLSc_bias}
			The shifted OLS estimator with intercept estimation $\hat{d}_{sOLS_c} \equiv h(n) \hat{d}_{OLS_c}$ is unbiased for
			\begin{align} 
				h(n) \equiv \frac{\sum_{i=1}^{n} \left[\ln (i/n)\right]^2 - n^{-1}\left[\sum_{i=1}^{n} \ln (i/n)\right]^2}{\sum_{i=1}^{n} \ln (i/n)-\sum_{j=1}^{n} j^{-1} \sum_{i=1}^j \ln (i/n)} . \notag
			\end{align}		
		\end{proposition}
		\begin{proof}
			From Equation \eqref{eq:d_OLSc} and given that $\E\left[ \ln \left(S_i / \underline{s}\right) \right] = d$ and $\E\left[ \ln \left(S_{(i)} / \underline{s}\right) \right] = \sum_{j=i}^n (d/j)$,
			\begin{align} 				
				\E\left(\hat{d}_{OLS_c}\right) = & \frac{\left[\sum_{i=1}^n  \ln (i/n) \right]\left[n^{-1} \sum_{i=1}^n \E\left(\ln(S_i/\underline{s})\right)\right]-\sum_{i=1}^n \ln (i/n) \E\left(\ln(S_{(i)}/\underline{s})\right)}{\sum_{i=1}^n \left[\ln (i/n)\right]^2-n^{-1} \left[\sum_{i=1}^n \ln (i/n) \right]^2} \notag \\
				\E\left(\hat{d}_{OLS_c}\right) = & \frac{\left[\sum_{i=1}^n  \ln (i/n) \right]d-\sum_{i=1}^n \ln (i/n) \left[\sum_{j=i}^n (d/j)\right]}{\sum_{i=1}^n \left[\ln (i/n)\right]^2-n^{-1} \left[\sum_{i=1}^n \ln (i/n) \right]^2} = \frac{d}{h(n)} . \notag
			\end{align}	
			As a consequence, $\E \left(\hat{d}_{sOLS_c} \right) = h(n) \E \left(\hat{d}_{OLS_c} \right) = d$.
		\end{proof}	
		\bigskip
		
		Figure \ref{fig:fn_gn_hn} plots the bias correction for the OLS \textit{with} intercept estimation, $h(n)$, comparing it to the bias corrections $f(n)$ and $g(n)$ for ML and OLS, respectively. As can be seen, the OLS adjustment is more severe when an intercept is estimated, meaning that $\hat{d}_{OLS_c}$ has higher bias than $\hat{d}_{OLS}$ in small samples.
		
		\begin{figure}[h!]
			\centering
			\includegraphics[scale=1]{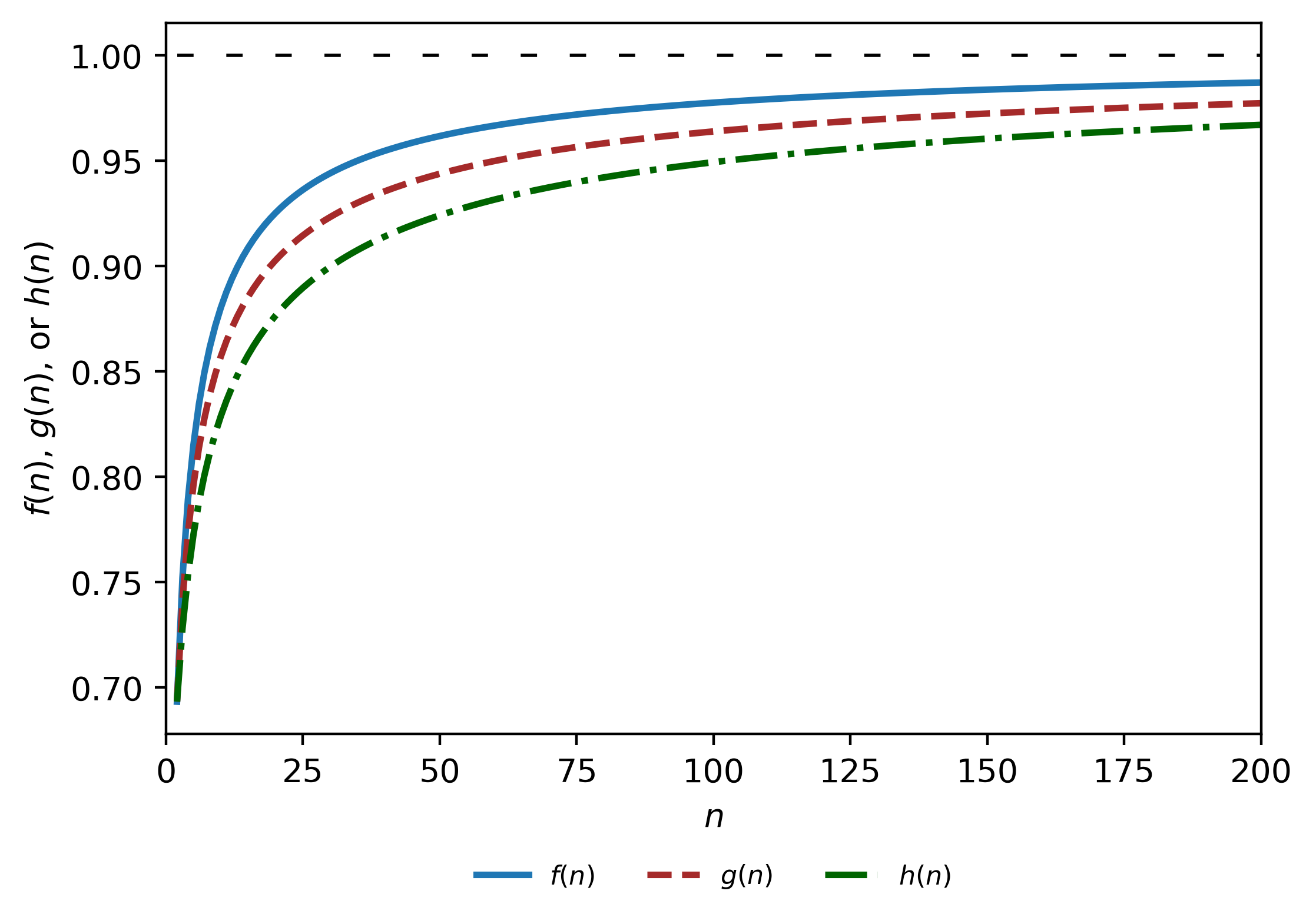}
			\caption{ML bias correction $f(n)$, OLS bias correction $g(n)$, and OLS with intercept estimation bias correction $h(n)$}
			\label{fig:fn_gn_hn}
		\end{figure}
		
		\subsection{Tail exponent estimation} \label{sec:app_formal_results_k}			
		
		Consider the problem of estimating the tail exponent $k = 1/d$. Naturally, the MLE of $k$ is $\hat{k}_{ML} =  1/\hat{d}_{ML}$, where $\hat{d}_{ML} = \frac{n-1}{n} \hat{d}_{sML}$ is the original, non-shifted, MLE from \cite{hill1975simple}. Let $\hat{k}_{OLS}$ be the OLS estimator from regressing $\ln \hat{P}(S\geq S_{(i)})=\ln (i/n)$ on $\ln \left(S_{(n)}/S_{(i)}\right)$. Formally, $\hat{k}_{OLS}$ is based on the regression equation
		\begin{align} 
			 \ln (i/n) = & k \ln \left(S_{(n)}/S_{(i)}\right)  + \psi_i \label{eq:reg_eq_k}
		\end{align}
		for $i=1,2,...,n-1$, where $\psi$ is an empirical error. Hence,
		\begin{align} 
			\hat{k}_{OLS} = & \frac{\sum_{i=1}^{n-1} \ln (i/n) \ln(S_{(n)}/S_{(i)})}{\sum_{i=1}^{n-1} \left[\ln (S_{(n)}/S_{(i)})\right]^2} . \label{eq:k_OLS}
		\end{align}			
		
		Given that, the following result holds:
		
		\begin{proposition} \label{prop:k_ML_WLS}
			Denote by $\hat{k}_{WLS}$ the WLS estimator of $k$ on \eqref{eq:reg_eq_k} with weights $\tilde{w}_i = 1/\ln \left(S_{(i)}/S_{(n)}\right)$ for $i=1,2,...,n-1$. Then, $\hat{k}_{ML} = \frac{n/(n-1)}{f(n)} \hat{k}_{WLS}$ .
		\end{proposition}
		\begin{proof}
			Applying WLS with weights $\tilde{w}_i = 1/\ln \left(S_{(i)}/S_{(n)}\right)$ on \eqref{eq:reg_eq_k},
			\begin{align} 
				\hat{k}_{WLS} = & \frac{\sum_{i=1}^{n-1} \tilde{w}_i \ln (i/n) \ln \left(S_{(n)}/S_{(i)}\right)}{\sum_{i=1}^{n-1} \tilde{w}_i \left[\ln \left(S_{(n)}/S_{(i)}\right)\right]^2} =\frac{-(n-1)^{-1} \sum_{i=1}^{n-1} \ln (i/n)}{(n-1)^{-1} \sum_{i=1}^{n-1} \ln\left(S_{(i)}/S_{(n)}\right)} = \frac{f(n)}{\hat{d}_{sML}} = \frac{f(n)}{\frac{n}{n-1} \hat{d}_{ML}} \notag \\
				\therefore \hat{k}_{ML} = & \frac{1}{\hat{d}_{ML}}  = \frac{n/(n-1)}{f(n)} \hat{k}_{WLS} , \notag
			\end{align}
			where, in the first line, we use $f(n) \equiv - (n-1)^{-1} \sum_{i=1}^{n-1} \ln (i/n)$, as defined in Proposition \ref{prop:d_ML_WLS}, $\hat{d}_{sML} = (n-1)^{-1} \sum_{i=1}^{n-1} \ln\left(S_{(i)}/S_{(n)}\right)$, from Equation \eqref{eq:d_sMLE}, and $\hat{d}_{ML} = \frac{n-1}{n} \hat{d}_{sML}$.
		\end{proof}	
		\bigskip
				
		Therefore, similar to the estimation of $d$, in large samples, the numerical difference between OLS and ML estimators is due to weighting, with the MLE giving more importance to smaller observations. Figure \ref{fig:k_weights} plots ML and OLS observations' weights normalized to add up to one for $n=3,10,100,10000$. Since the weights are random in this case, the results shown are the average over 10,000 samples drawn from a Pareto distribution with $k=1$ and $\underline{s} = 50$.\footnote{In any case, the results do not seem to be sensitive to these parameters.} Similar to what we have seen for $d$ in Figure \ref{fig:d_weights}, the MLE of $k$ assigns significantly greater weight to observations at begging of the support, especially for larger samples. Once more, for $n=10000$, the 20\% smallest observations have roughly 90\% of the weight!
										
		\begin{figure}[h!]
			\centering
			\includegraphics[scale=1]{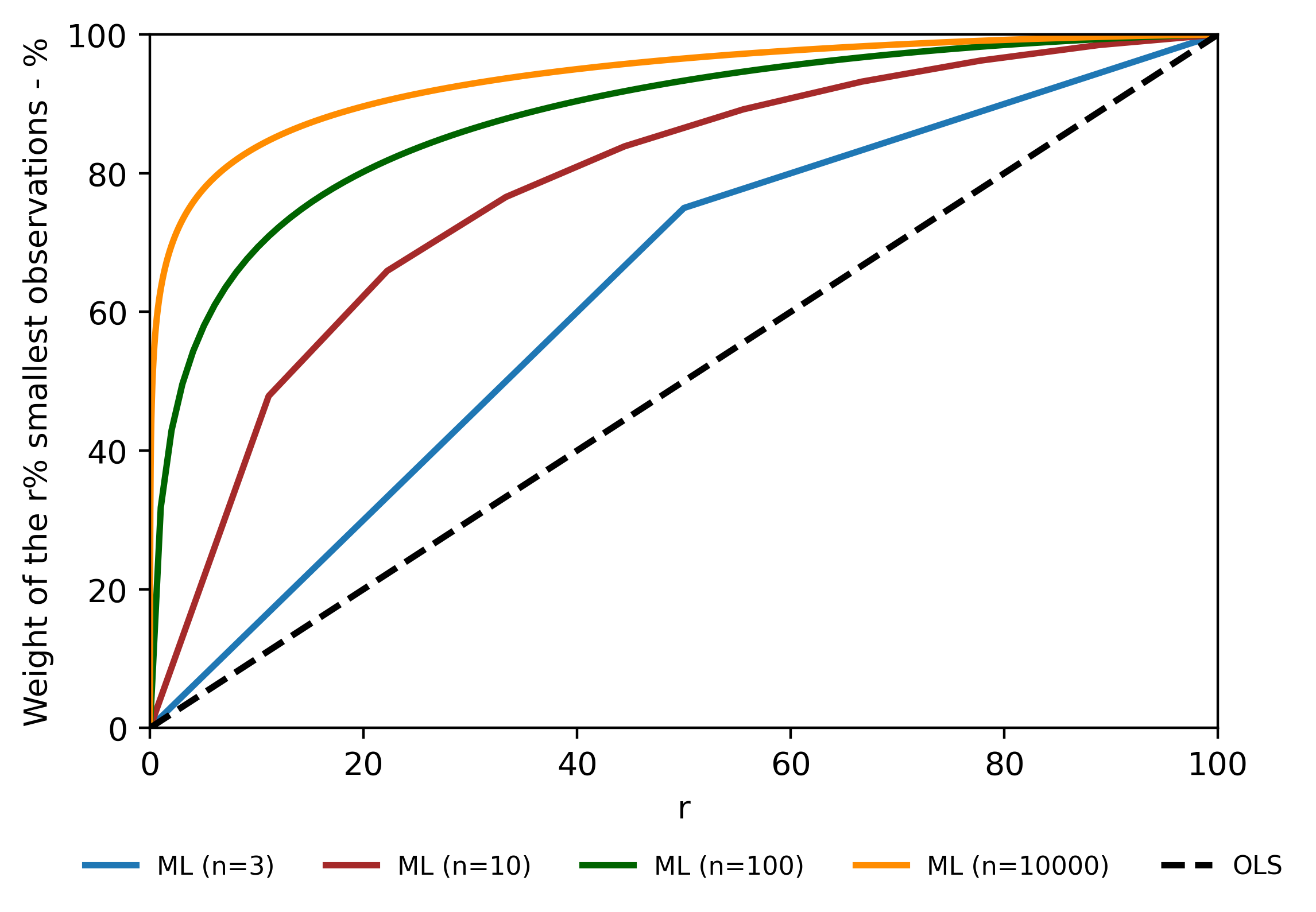}
			\caption{Observations' weights for $k$ estimators: ML \textit{versus} OLS}
			\label{fig:k_weights}
		\end{figure}
	
		\clearpage
		\section{Monte Carlo exercises using other exponents} \label{sec:app_mc}
		
		\subsection{Strictly Pareto but only in the upper tail} \label{sec:app_mc1}
				
		\begin{figure}[h!]
			\centering
			\includegraphics[scale=0.84]{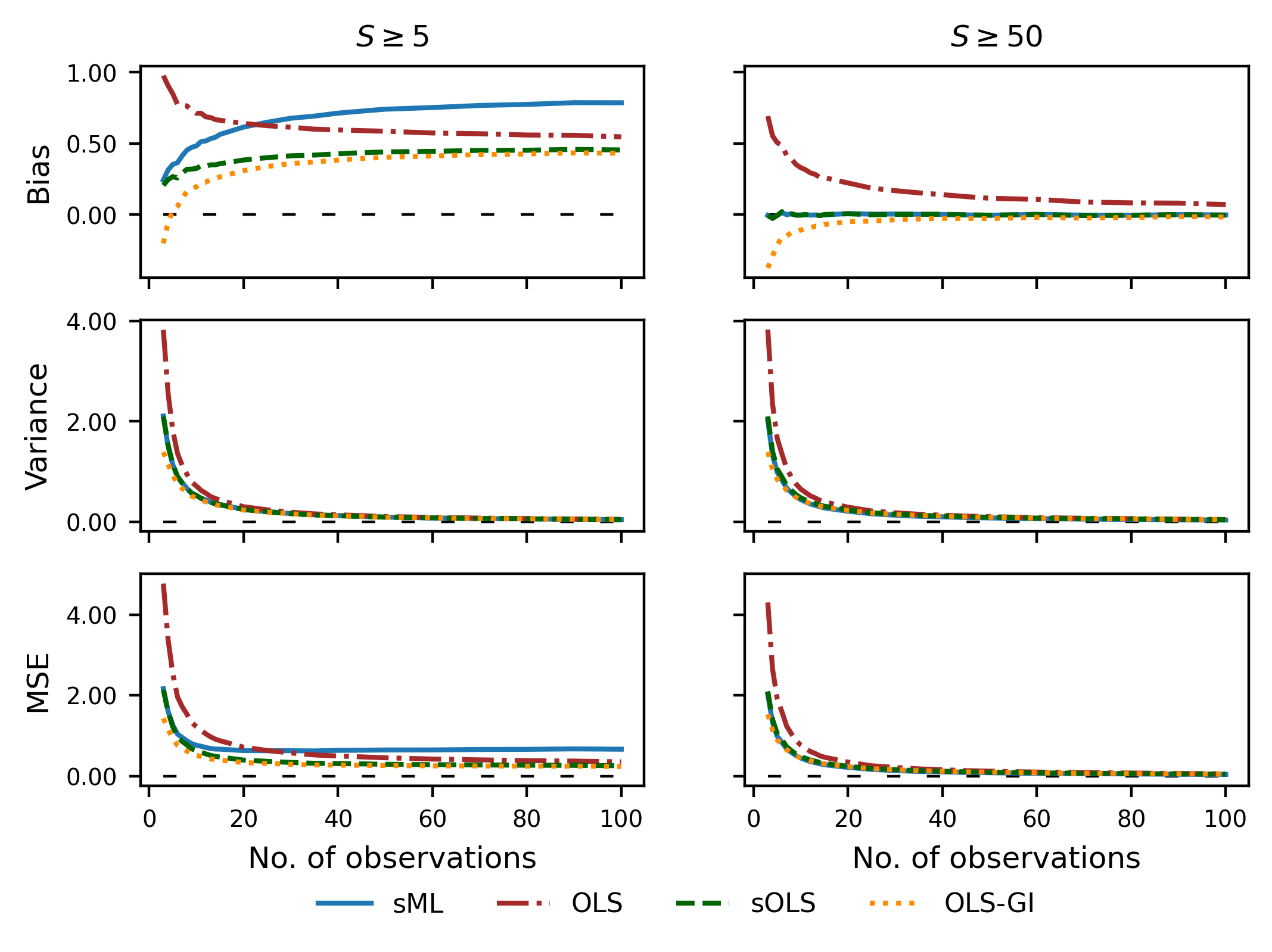}
			\caption{Performance of $d$ estimators in the first Monte Carlo exercise, $k=1/2$}
		\end{figure}
				
		\begin{figure}[h!]
			\centering
			\includegraphics[scale=0.84]{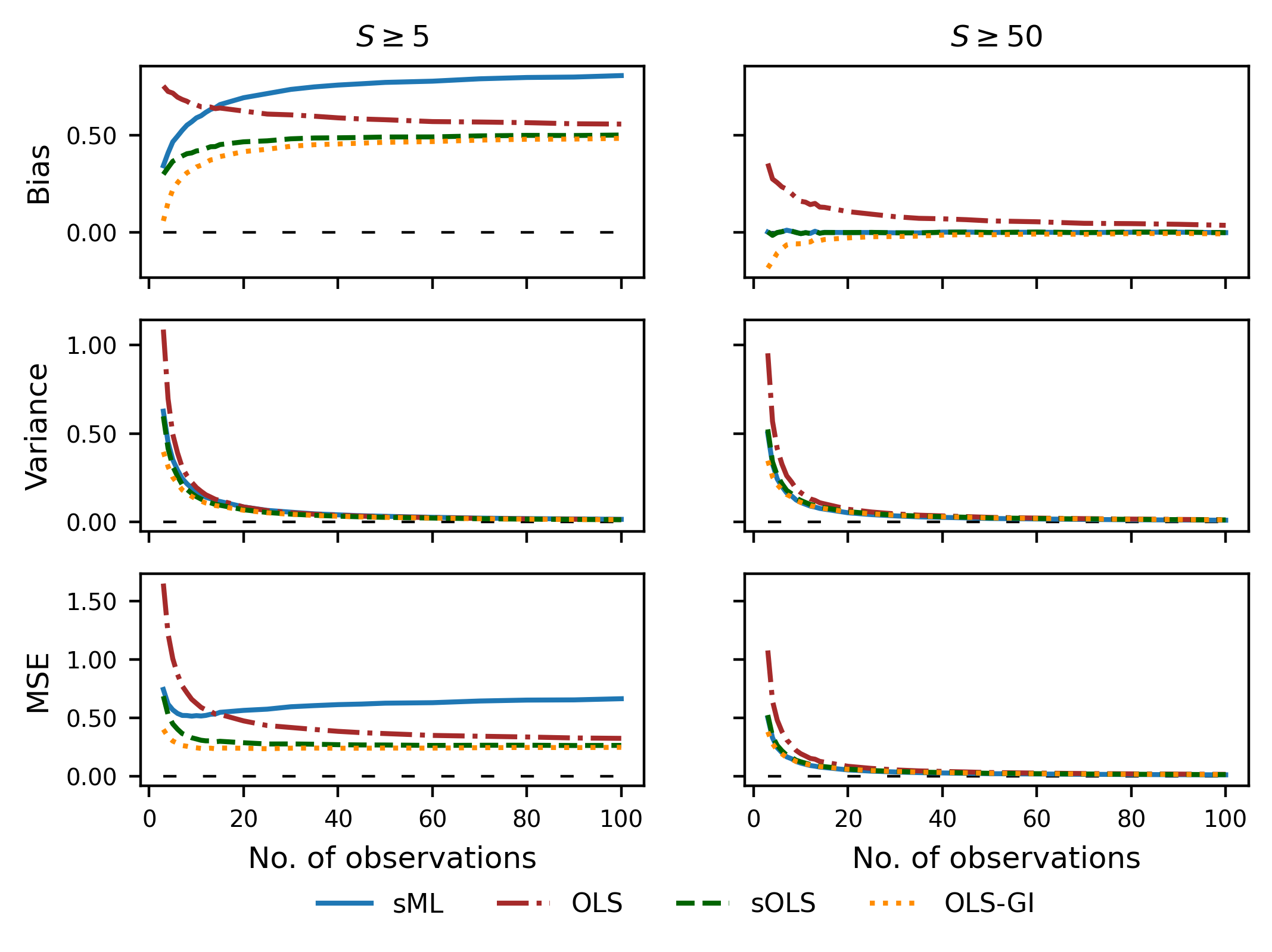}
			\caption{Performance of $d$ estimators in the first Monte Carlo exercise, $k=1$}
		\end{figure}
			
		\begin{figure}[h!]
			\centering
			\includegraphics[scale=0.84]{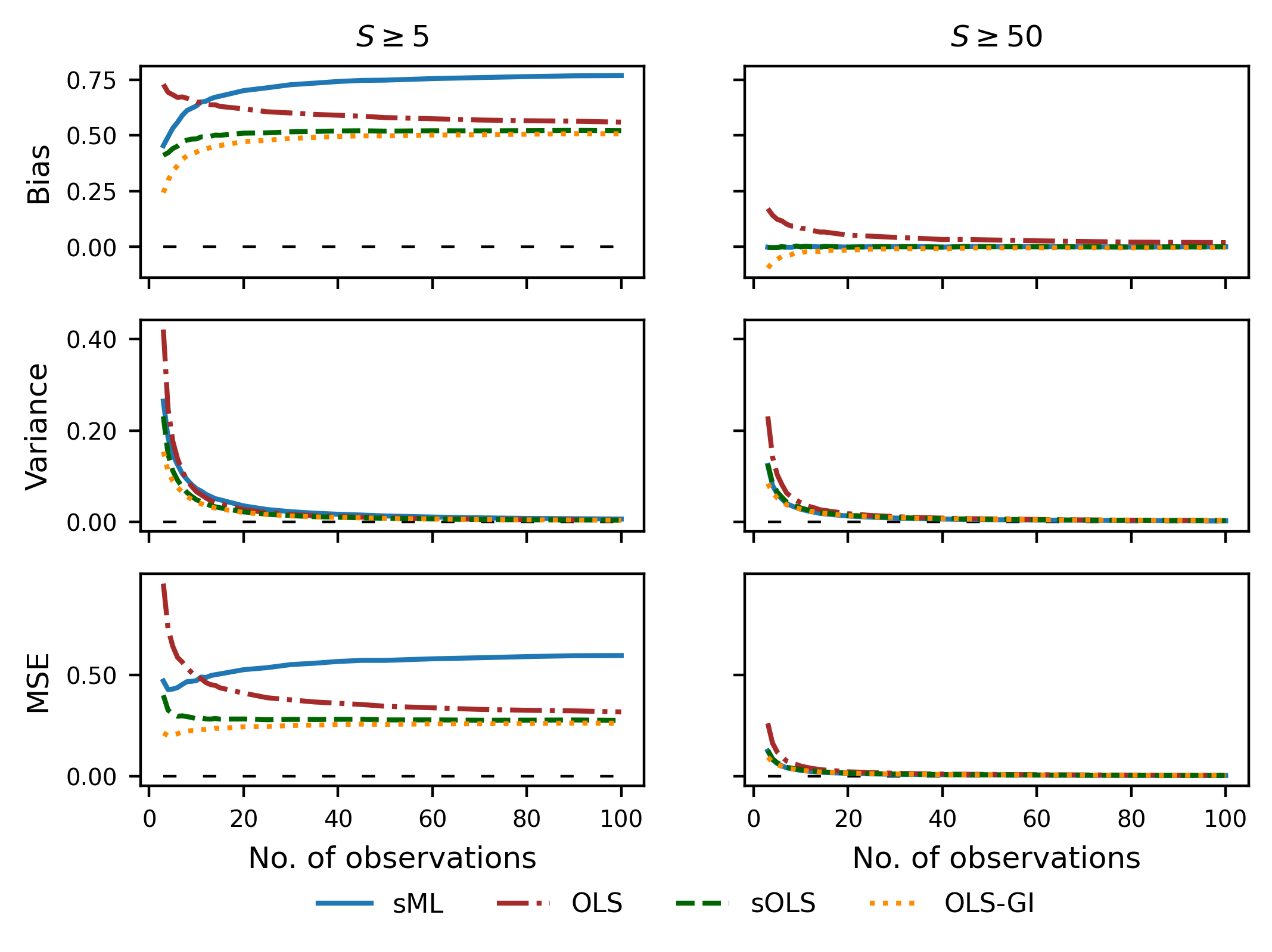}
			\caption{Performance of $d$ estimators in the first Monte Carlo exercise, $k=2$}
		\end{figure}
			
		\begin{figure}[h!]
			\centering
			\includegraphics[scale=0.84]{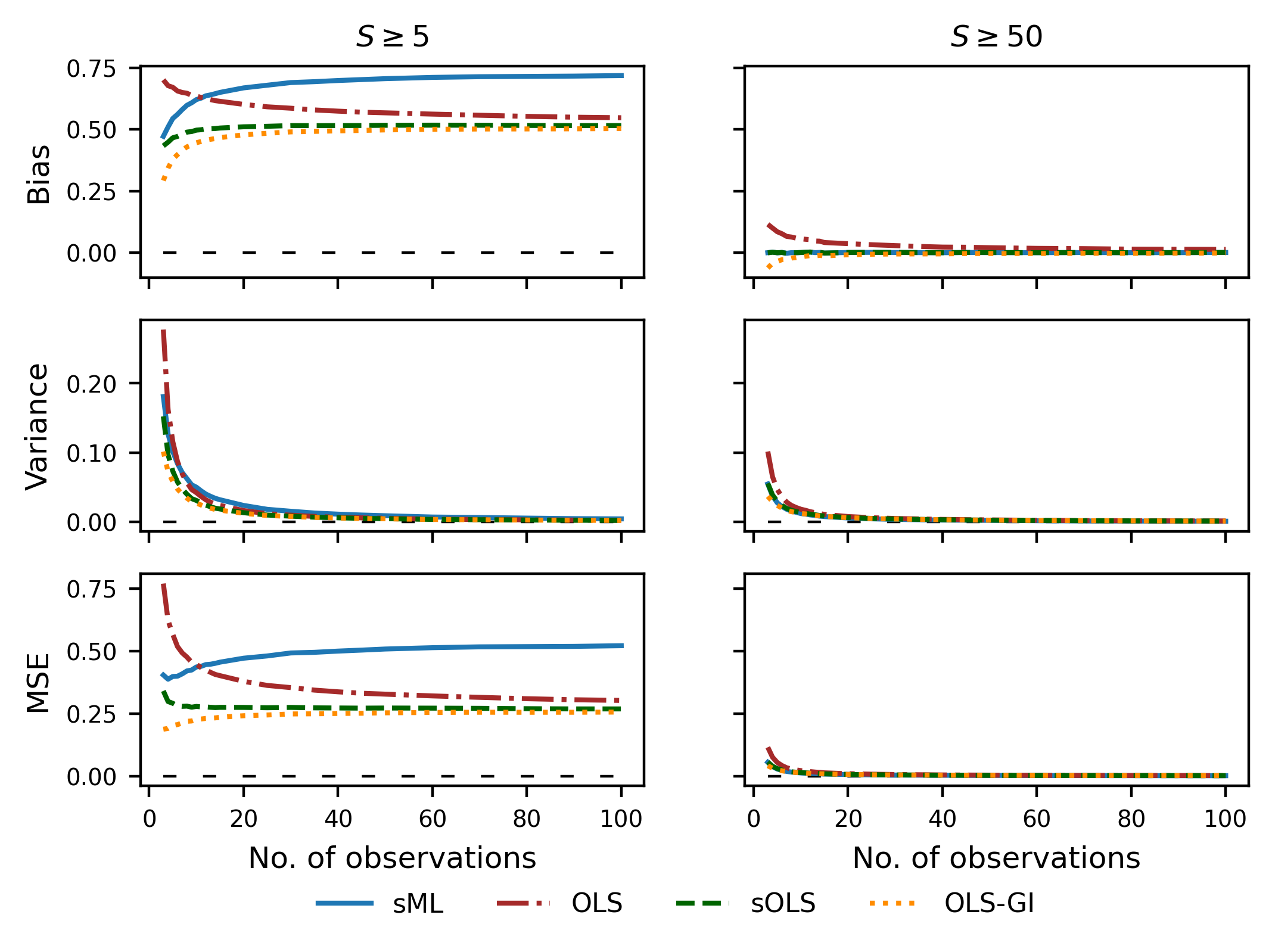}
			\caption{Performance of $d$ estimators in the first Monte Carlo exercise, $k=3$}
		\end{figure}

		\clearpage
		\subsection{Regularly varying} \label{sec:app_mc2}
			
		\begin{figure}[h!]
			\centering
			\includegraphics[scale=0.84]{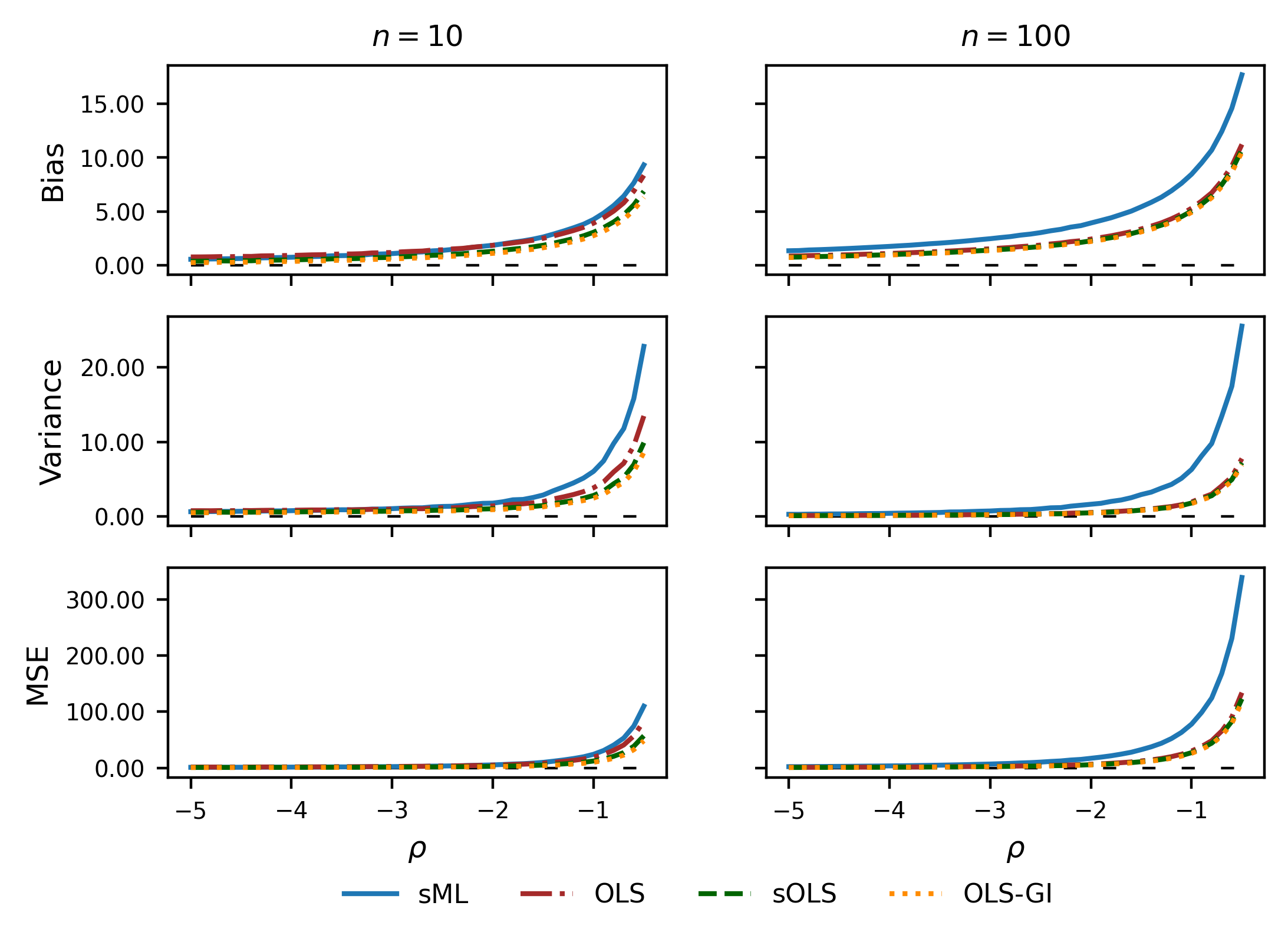}
			\caption{Performance of $d$ estimators in the second Monte Carlo exercise, $k=1/2$}
		\end{figure}
			
		\begin{figure}[h!]
			\centering
			\includegraphics[scale=0.84]{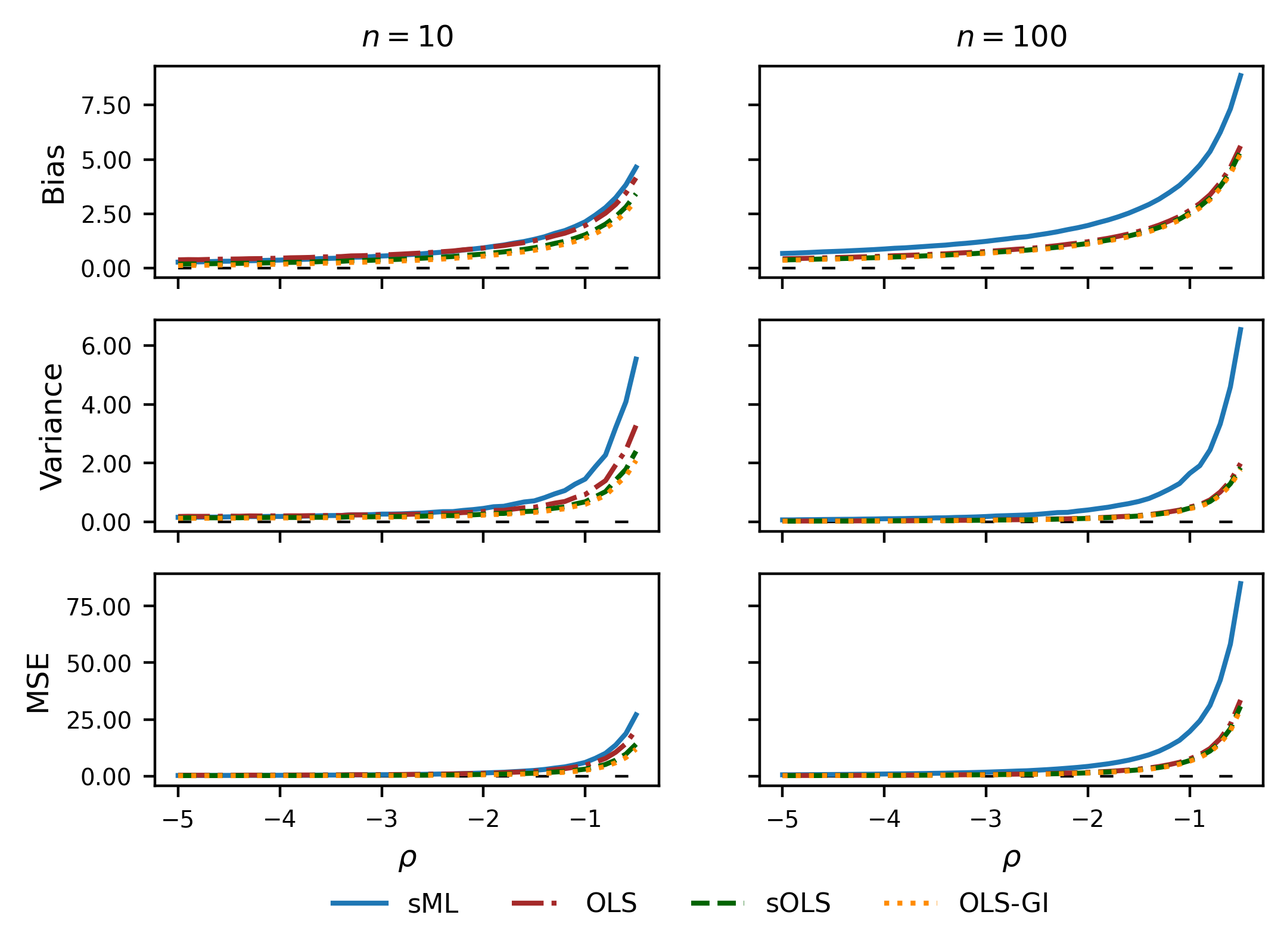}
			\caption{Performance of $d$ estimators in the second Monte Carlo exercise, $k=1$}
		\end{figure}
		
		\begin{figure}[h!]
			\centering
			\includegraphics[scale=0.84]{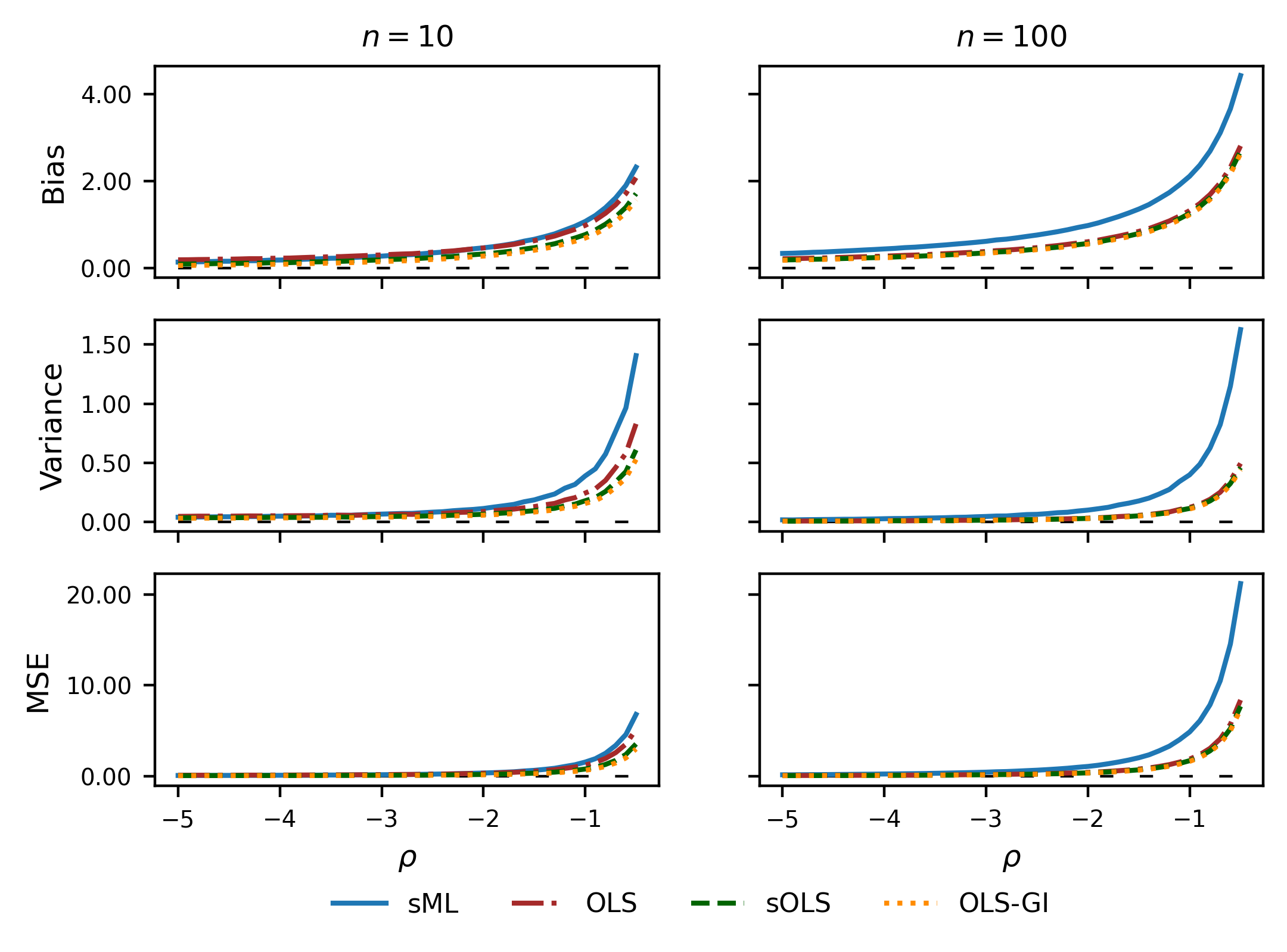}
			\caption{Performance of $d$ estimators in the second Monte Carlo exercise, $k=2$}
		\end{figure}
			
		\begin{figure}[h!]
			\centering
			\includegraphics[scale=0.84]{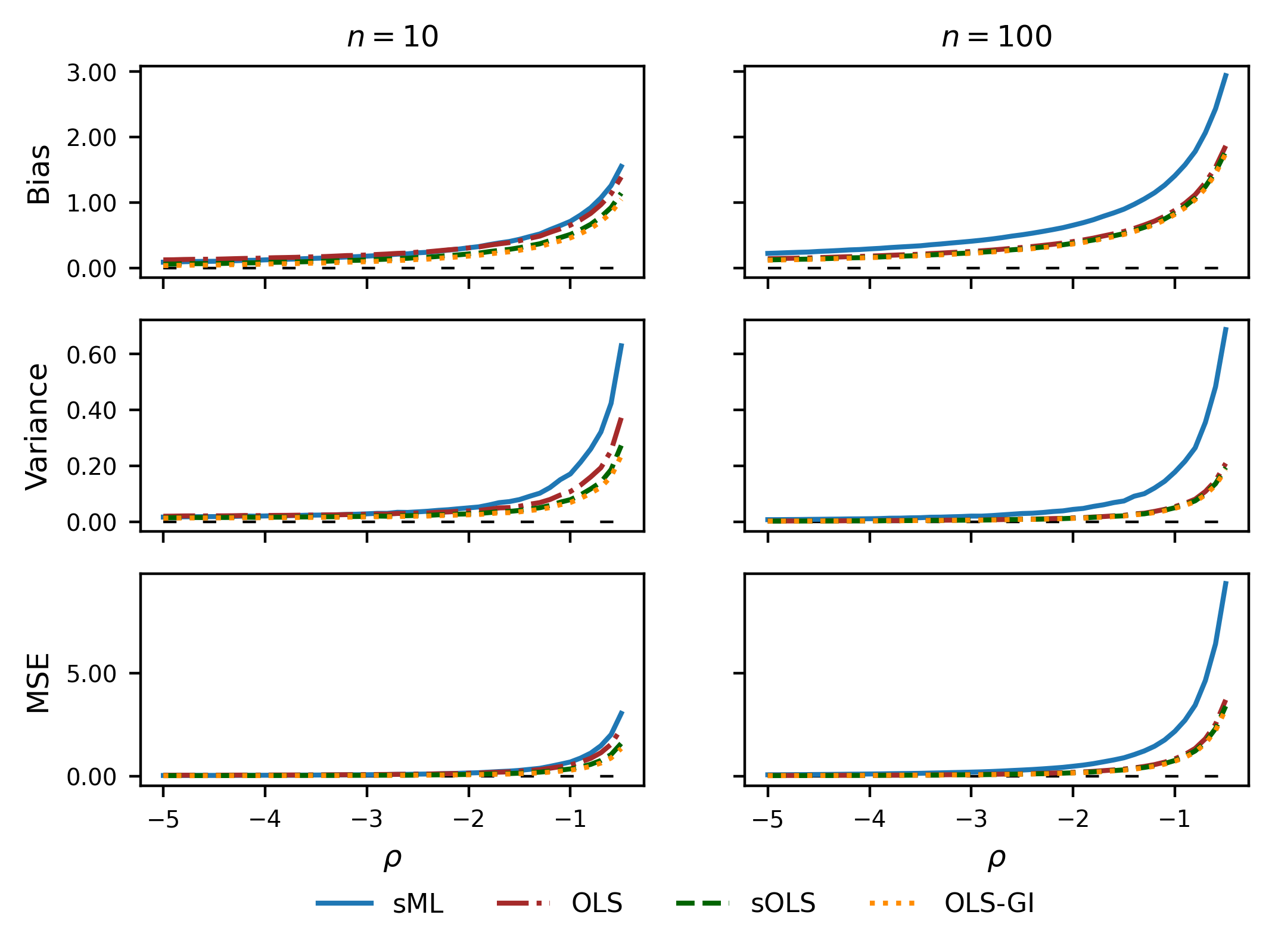}
			\caption{Performance of $d$ estimators in the second Monte Carlo exercise, $k=3$}
		\end{figure}
			
		\clearpage
		\section{Empirical applications using all possible supports} \label{sec:app_emp_exs}
				
		\subsection{Distribution of US cities by population} \label{sec:app_emp_exs_UScities}
				
		\begin{figure}[h!]
			\centering
			\includegraphics[scale=0.76]{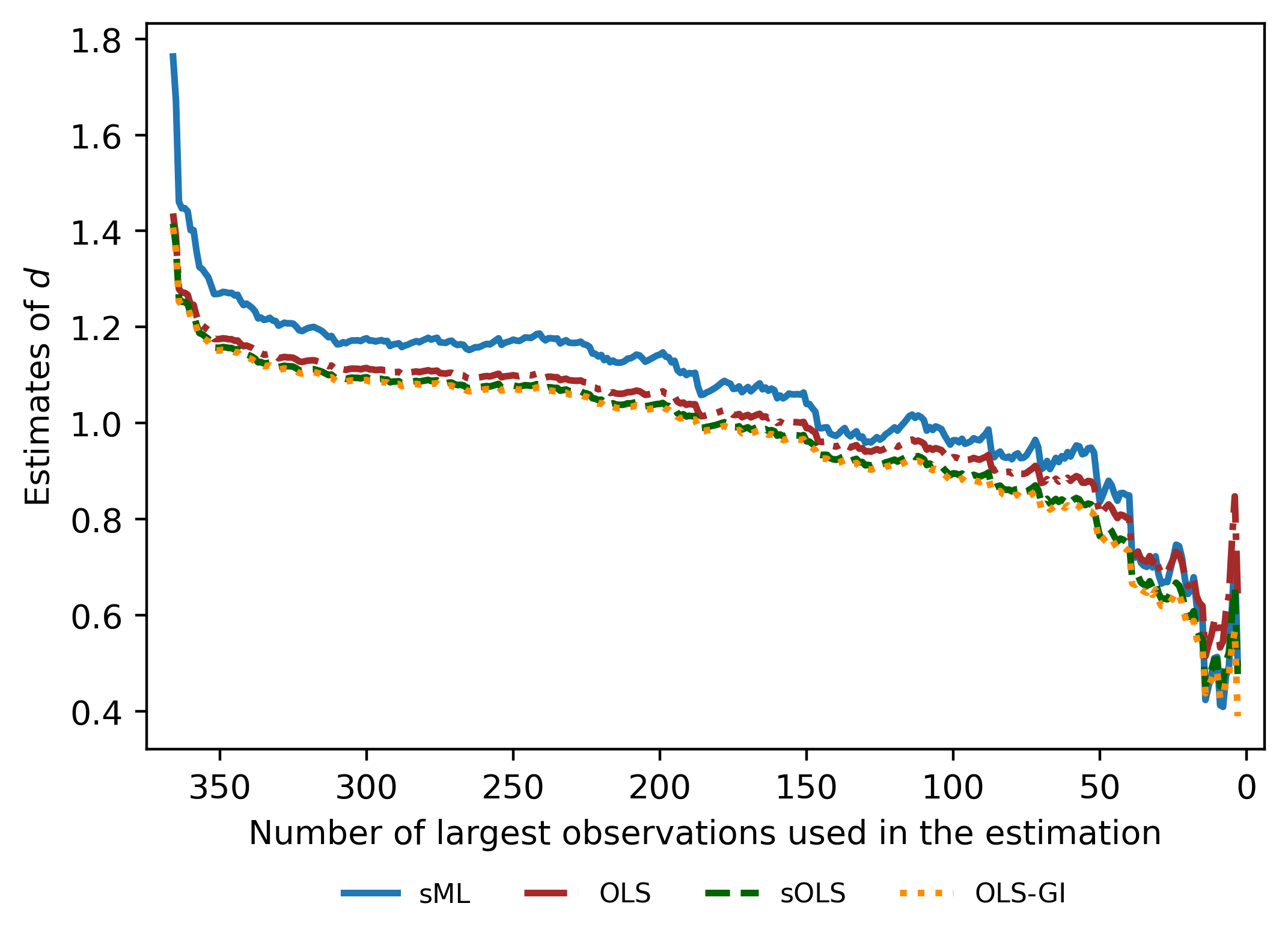}
			\caption{Estimates of $d$ for US cities' populations, all possible supports}
			\label{fig:d_UScities_all}
		\end{figure}
			
		\subsection{Distribution of S\&P absolute daily returns} \label{sec:app_emp_exs_S&Pr}

		\begin{figure}[h!]
			\centering
			\includegraphics[scale=0.76]{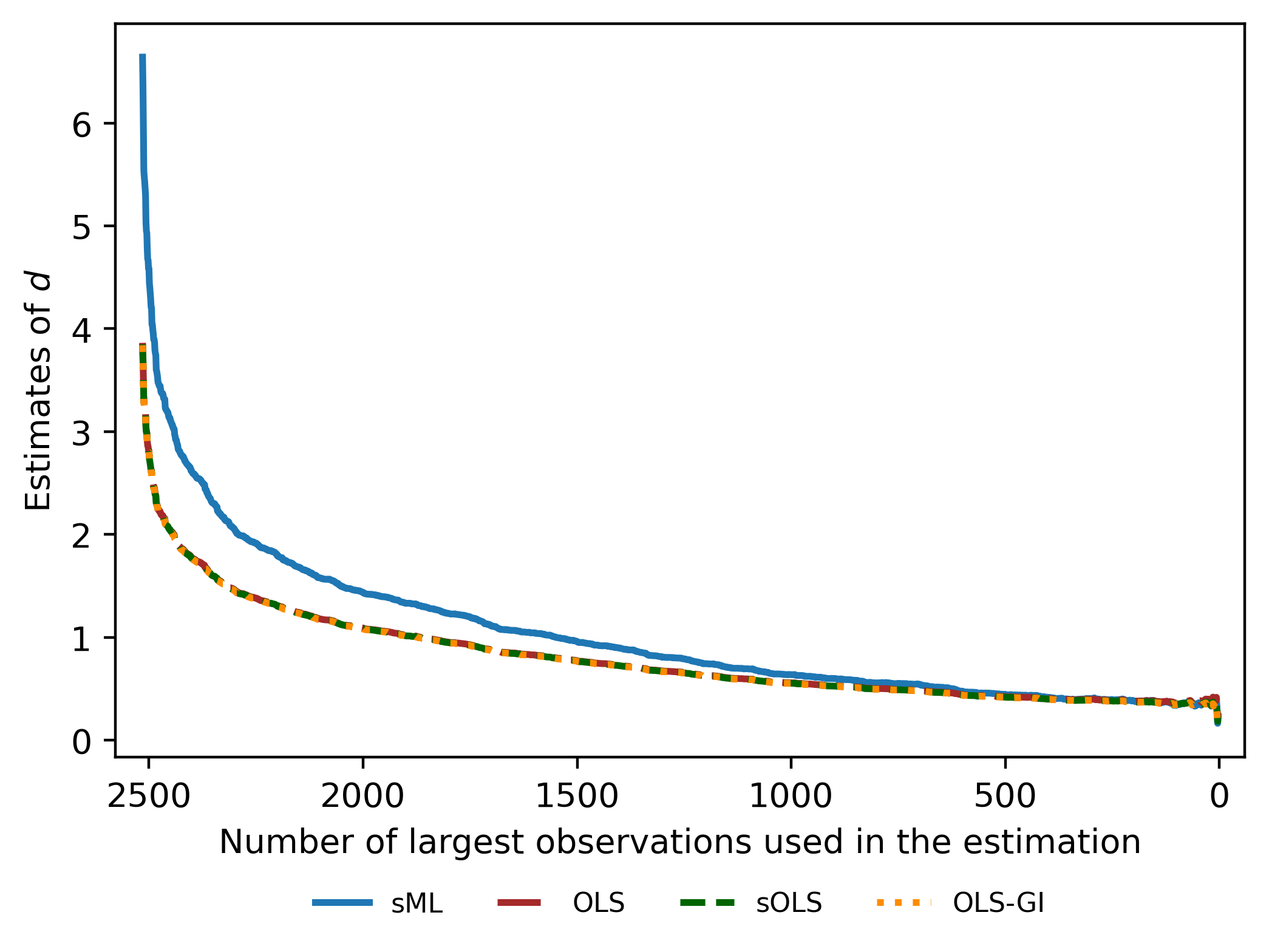}
			\caption{Estimates of $d$ for S\&P absolute daily returns, all possible supports}
			\label{fig:d_SPr_all}
		\end{figure}

	\end{appendices}
	
	\clearpage

\bibliographystyle{agsm}

\bibliography{bibliografia}

@article{gabaix2009power,
  title={Power laws in economics and finance},
  author={Gabaix, Xavier},
  journal={Annual Review of Economics},
  volume={1},
  number={1},
  pages={255--293},
  year={2009},
  publisher={Annual Reviews}
}

@article{axtell2001zipf,
  title={Zipf distribution of {U.S.} firm sizes},
  author={Axtell, Robert L},
  journal={Science},
  volume={293},
  number={5536},
  pages={1818--1820},
  year={2001},
  publisher={American Association for the Advancement of Science}
}

@article{gabaix_landier2008why,
  title={Why has {CEO} pay increased so much?},
  author={Gabaix, Xavier and Landier, Augustin},
  journal={The Quarterly Journal of Economics},
  volume={123},
  number={1},
  pages={49--100},
  year={2008},
  publisher={Oxford University Press}
}

@article{luttmer2007selection,
  title={Selection, growth, and the size distribution of firms},
  author={Luttmer, Erzo G J},
  journal={The Quarterly Journal of Economics},
  volume={122},
  number={3},
  pages={1103--1144},
  year={2007},
  publisher={Oxford University Press}
}

@article{okuyama_etal1999zipf,
  title={Zipf's law in income distribution of companies},
  author={Okuyama, Kazumi and Takayasu, Misako and Takayasu, Hideki},
  journal={Physica A: Statistical Mechanics and its Applications},
  volume={269},
  number={1},
  pages={125--131},
  year={1999},
  publisher={Elsevier}
}

@article{gabaix2016power,
  title={Power laws in economics: An introduction},
  author={Gabaix, Xavier},
  journal={Journal of Economic Perspectives},
  volume={30},
  number={1},
  pages={185--206},
  year={2016},
  publisher={American Economic Association}
}

@book{zipf1949human,
  title={Human behavior and the principle of least effort: An introduction to human ecology},
  author={Zipf, George Kingsley},
  year={1949},
  publisher={Addison-Wesley Press}
}

@article{gabaix1999zipf,
  title={Zipf's Law for Cities: An Explanation},
  author={Gabaix, Xavier},
  journal={The Quarterly Journal of Economics},
  volume={114},
  number={3},
  pages={739--767},
  year={1999},
  publisher={Oxford University Press}
}

@inbook{gabaix_ioannides2004evolution,
  author={Gabaix, Xavier and Ioannides, Yannis M},
  title={The Evolution of City Size Distributions},
  booktitle={Handbook of Regional and Urban Economics},
  volume={4},
  chapter={53},
  year={2004},
  pages={2341-2378},
  publisher={North-Holland},
  organization={North-Holland},
}

@article{hill1975simple,
  title={A simple general approach to inference about the tail of a distribution},
  author={Hill, Bruce M},
  journal={The Annals of Statistics},
  volume={3},
  number={5},
  pages={1163--1174},
  year={1975},
  publisher={JSTOR}
}

@article{gabaix_ibragimov2011rank,
  title={Rank-1/2: A simple way to improve the {OLS} estimation of tail exponents},
  author={Gabaix, Xavier and Ibragimov, Rustam},
  journal={Journal of Business \& Economic Statistics},
  volume={29},
  number={1},
  pages={24--39},
  year={2011},
  publisher={Taylor \& Francis}
}

@article{digiovanni_etal2011power,
  title={Power laws in firm size and openness to trade: Measurement and implications},
  author={Di Giovanni, Julian and Levchenko, Andrei A and Ranciere, Romain},
  journal={Journal of International Economics},
  volume={85},
  number={1},
  pages={42--52},
  year={2011},
  publisher={Elsevier}
}

@article{fujiwara_etal2004pareto,
  title={Do {P}areto--{Z}ipf and {G}ibrat laws hold true? An analysis with {E}uropean firms},
  author={Fujiwara, Yoshi and Di Guilmi, Corrado and Aoyama, Hideaki and Gallegati, Mauro and Souma, Wataru},
  journal={Physica A: Statistical Mechanics and its Applications},
  volume={335},
  number={1-2},
  pages={197--216},
  year={2004},
  publisher={Elsevier}
}

@article{digiovanni_levchenko2013firm,
  title={Firm entry, trade, and welfare in Zipf's world},
  author={Di Giovanni, Julian and Levchenko, Andrei A},
  journal={Journal of International Economics},
  volume={89},
  number={2},
  pages={283--296},
  year={2013},
  publisher={Elsevier}
}

@article{clauset_etal2009power,
  title={Power-law distributions in empirical data},
  author={Clauset, Aaron and Shalizi, Cosma Rohilla and Newman, Mark EJ},
  journal={SIAM Review},
  volume={51},
  number={4},
  pages={661--703},
  year={2009},
  publisher={SIAM}
}

@article{schluter2018top,
  title={Top incomes, heavy tails, and rank-size regressions},
  author={Schluter, Christian},
  journal={Econometrics},
  volume={6},
  number={1},
  article-number={10},
  pages={1--16},
  year={2018},
  publisher={MDPI}
}

@article{urzua2011testing,
  title={Testing for {Z}ipf’s law: A common pitfall},
  author={Urz{\'u}a, Carlos M},
  journal={Economics Letters},
  volume={112},
  number={3},
  pages={254--255},
  year={2011},
  publisher={Elsevier}
}

@article{aban_meerschaert2004generalized,
  title={Generalized least-squares estimators for the thickness of heavy tails},
  author={Aban, Inmaculada B and Meerschaert, Mark M},
  journal={Journal of Statistical Planning and Inference},
  volume={119},
  number={2},
  pages={341--352},
  year={2004},
  publisher={Elsevier}
}

@article{hall1982some,
  title={On some simple estimates of an exponent of regular variation},
  author={Hall, Peter},
  journal={Journal of the Royal Statistical Society: Series B (Methodological)},
  volume={44},
  number={1},
  pages={37--42},
  year={1982},
  publisher={Wiley Online Library}
}

@book{pareto1896cours,
  title={Cours d'{\'e}conomie politique},
  author={Pareto, Vilfredo},
  volume={1},
  year={1896},
  publisher={Librairie Droz}
}

@article{gopikrishnan_etal1999scaling,
  title={Scaling of the distribution of fluctuations of financial market indices},
  author={Gopikrishnan, Parameswaran and Plerou, Vasiliki and Amaral, Luis A Nunes and Meyer, Martin and Stanley, H Eugene},
  journal={Physical Review E},
  volume={60},
  number={5},
  pages={5305--5316},
  year={1999},
  publisher={APS}
}

@article{gopikrishnan_etal2000statistical,
  title={Statistical properties of share volume traded in financial markets},
  author={Gopikrishnan, Parameswaran and Plerou, Vasiliki and Gabaix, Xavier and Stanley, H Eugene},
  journal={Physical Review E},
  volume={62},
  number={4},
  pages={R4493--R4496},
  year={2000},
  publisher={APS}
}

@article{klass_etal2006forbes,
  title={The {F}orbes 400 and the {P}areto wealth distribution},
  author={Klass, Oren S and Biham, Ofer and Levy, Moshe and Malcai, Ofer and Solomon, Sorin},
  journal={Economics Letters},
  volume={90},
  number={2},
  pages={290--295},
  year={2006},
  publisher={Elsevier}
}

@article{kratz_resnick1996qq,
  title={The {QQ}-estimator and heavy tails},
  author={Kratz, Marie and Resnick, Sidney I},
  journal={Stochastic Models},
  volume={12},
  number={4},
  pages={699--724},
  year={1996},
  publisher={Taylor \& Francis}
}

@article{csorgo_viharos1997asymptotic,
  title={Asymptotic normality of least-squares estimators of tail indices},
  author={Cs{\"o}rg{\"o}, S{\'a}ndor and Viharos, L{\'a}szl{\'o}},
  journal={Bernoulli},
  volume={3},
  number={3},
  pages={351--370},
  year={1997},
  publisher={International Statistical Institute (ISI)}
}

@article{beirlant_etal1996tail,
  title={Tail index estimation, {P}areto quantile plots, and regression diagnostics},
  author={Beirlant, Jan and Vynckier, Petra and Teugels, Jozef L},
  journal={Journal of the American Statistical Association},
  volume={91},
  number={436},
  pages={1659--1667},
  year={1996},
  publisher={Taylor \& Francis}
}

@article{breiman_etal1990robust,
  title={Robust confidence bounds for extreme upper quantiles},
  author={Breiman, Leo and Stone, Charles J and Kooperberg, Charles},
  journal={Journal of Statistical Computation and Simulation},
  volume={37},
  number={3-4},
  pages={127--149},
  year={1990},
  publisher={Taylor \& Francis}
}

@article{dekkers_dehaan1993optimal,
  title={Optimal choice of sample fraction in extreme-value estimation},
  author={Dekkers, Arnold LM and Dehaan, Laurens},
  journal={Journal of Multivariate Analysis},
  volume={47},
  number={2},
  pages={173--195},
  year={1993},
  publisher={Elsevier}
}

@article{drees_kaufmann1998selecting,
  title={Selecting the optimal sample fraction in univariate extreme value estimation},
  author={Drees, Holger and Kaufmann, Edgar},
  journal={Stochastic Processes and their Applications},
  volume={75},
  number={2},
  pages={149--172},
  year={1998},
  publisher={Elsevier}
}

@article{clauset_etal2007frequency,
  title={On the frequency of severe terrorist events},
  author={Clauset, Aaron and Young, Maxwell and Gleditsch, Kristian Skrede},
  journal={Journal of Conflict Resolution},
  volume={51},
  number={1},
  pages={58--87},
  year={2007},
  publisher={Sage Publications}
}

@article{danielsson_etal2001using,
  title={Using a bootstrap method to choose the sample fraction in tail index estimation},
  author={Danielsson, Jon and de Haan, Laurens and Peng, Liang and de Vries, Casper G},
  journal={Journal of Multivariate Analysis},
  volume={76},
  number={2},
  pages={226--248},
  year={2001},
  publisher={Elsevier}
}

@article{handcock_jones2004likelihood,
  title={Likelihood-based inference for stochastic models of sexual network formation},
  author={Handcock, Mark S and Jones, James Holland},
  journal={Theoretical Population Biology},
  volume={65},
  number={4},
  pages={413--422},
  year={2004},
  publisher={Elsevier}
}
\end{document}